%% file: main.tex
\documentclass[journal,onecolumn,12pt,table]{IEEEtran}
\renewcommand{\baselinestretch}{1.5}
\IEEEoverridecommandlockouts
\input{preamble}
\input{preamble_figure}

\newcommand{\Pb}[1]{\Pr\left(#1\right)}
\newcommand{\pb}[1]{p\left(#1\right)}
\newcommand{\qb}[1]{q\left(#1\right)}
\newcommand{\wb}[1]{w\left(#1\right)}

\begin{document}
\title{ON-OFF Privacy in the Presence of Correlation} 

\author{%
  \IEEEauthorblockN{Fangwei Ye\IEEEauthorrefmark{1}, Carolina Naim\IEEEauthorrefmark{2}, Salim El Rouayheb\IEEEauthorrefmark{2}%
  \thanks{A preliminary version of this paper was presented at IEEE International Symposium on Information Theory, Los Angeles, CA, USA, 2020. 
} 
}

  \IEEEauthorblockA{\IEEEauthorrefmark{1}%
Broad Institute of MIT and Harvard,                   
              Email: fye@broadinstitute.org} 

                \and

                \IEEEauthorblockA{\IEEEauthorrefmark{2}%
   Rutgers University, Email: \{carolina.naim, salim.elrouayheb\}@rutgers.edu}

}

\maketitle

\begin{abstract}

We formulate and study the problem of  ON-OFF privacy. ON-OFF privacy algorithms enable a user to continuously switch   his privacy   between  ON and OFF.  An obvious example  is the incognito mode in internet browsers. But beyond internet browsing, ON-OFF privacy   can be a desired feature in most online applications. The challenge is that  the statistical correlation over time of a user's online behavior can lead to leakage of information.

We consider  the setting in which a user is interested in retrieving the latest message generated by one of $N$ sources. The user's privacy status can change between ON and OFF over time. When privacy is ON the user wants to hide his request. Moreover, since the user's requests depend on personal attributes such as age, gender, and political views, they are typically correlated over time. As a consequence,  the user cannot simply ignore privacy when privacy is OFF. We  model the correlation between user's requests by an $N$ state Markov chain. The goal is to design query schemes with optimal download rate, that preserve privacy in an ON-OFF privacy setting. In this paper, we present inner and outer bounds on the achievable download rate for $N$ sources.  
 We also devise an efficient algorithm to construct an ON-OFF privacy scheme achieving the inner bound and prove its optimality in the case 
 $N=2$ sources. For $N > 2$, finding tighter outer bounds and efficient constructions of ON-OFF privacy schemes that would achieve them remains an open question.


\end{abstract}

\section{Introduction}
\label{sec:introduction}

  Privacy is a major concern for online users who can unknowingly reveal critical personal information (age, sex, diseases, political proclivity, etc.) through daily online activities such as watching online videos, following people and liking posts on social media, reading news, and searching websites. This is a well-acknowledged concern and has lead to many interesting theoretical problems such as anonymity \cite{Sweeney_2002}, differential privacy \cite{Dwork_2006}, private information retrieval \cite{Chor_1995}, and other privacy-preserving algorithms.

 
The implicit assumption that is common in   existing privacy models is that   the user wants  privacy {\em all the time}.
 We    refer to it    as privacy being always ON. However, privacy-preserving algorithms  incur high    costs   on the service provider, and can lead to degraded quality of service  at the user's side. One should  think of privacy as an expensive utility,  which should be turned ON only when   needed (depending on geographical  location, device, network, etc.). This motivated us to introduce  and study the problem of   ON-OFF privacy~\cite{Naim_2019}.   ON-OFF privacy algorithms enable a user to   switch   his/her privacy   between  ON and OFF.  A current   application that allows to switch between a private and a non-private mode is     internet browsers. But beyond internet browsing, ON-OFF privacy   can be a desired feature in many online applications.


One may be tempted to propose the simple solution in which the user has available to him two schemes, one private and one non-private. Over time, the user simply switches between these two schemes depending on whether privacy is turned ON or OFF. The problem with this solution is that it  guarantees privacy only if the user's online activities are statistically independent over time. However, a user's online activities are typically personal, making them correlated over time. For example, a bilingual English/Spanish user, who is  checking  the news in Spanish now, is more likely to keep reading the news in Spanish for a while before switching to English. At that point English becomes more probable.  Another example is when  the user   is watching online videos. One may think of a scenario where the user is more likely to watch the top item from a  list of recommended videos  that depends on the previously watched videos.
Thus,  due to correlation, simply ignoring the privacy requirement when privacy is OFF may reveal information about the activities when privacy was ON.  Location based services are another example that can benefit from ON-OFF privacy algorithms. Imagine a  user who does not care about revealing his/her location   right now, but wants to hide it  a minute ago. He/she still has to be careful not to completely reveal his/her current location because it will leak information about where he/she was a minute ago.  

\subsection{Example}
\label{sec:example}
To be more concrete and to  gently introduce our setup for ON-OFF privacy, we give the following example.  Suppose a user is watching political or news videos online. At each time $t$, the user has a choice between two new videos each of which is produced by two different news sources, $A$ or $B$.  Source $A$  is politically left-leaning and source $B$ is right-leaning.  

Let $X_t\in\{A,B\}$ be the source whose video the user wants to watch at time $t \in \mathbb{N}$. We model the correlation among the user's requests by assuming that $X_t$ is the two-state Markov chain  depicted in Figure~\ref{fig:example}, where the transition probabilities are given by $\alpha=\Pr(X_{t+1}=B\mid X_t=A)$ and 
$\beta=\Pr(X_{t+1}=A\mid X_t=B)$.
For illustration, we choose $\alpha=\beta=0.2$. This means that if the current video being watched is left-leaning, there is an $80\%$ chance that the next video is also left-leaning, and vice versa. 

\begin{figure}[b]
\vspace{-1pt}
\centering
\begin{tikzpicture}[thick,scale=0.8, every node/.style={scale=0.8}]
 \tikzstyle{every node}=[font=\normalsize]
        \tikzset{node style/.style={state, 
        							inner sep=0.5pt,
                                    minimum size = 25pt,
                                    line width=0.2mm,
                                    fill=white},
                    LabelStyle/.style = { minimum width = 1em, fill = white!10,
                                            text = black},
                    EdgeStyle/.append style = {->, bend left=25, line width=0.2mm} }

        \node[node style] at (4,7)     (pro-left)     {A};
        \node[node style] at (8, 7)     (pro-right)     {B};
		  \Edge[label = $\alpha$](pro-left)(pro-right)
  			\Edge[label = $\beta$](pro-right)(pro-left)
            
            \Loop[dist = 2.5cm, dir = NO, label = $1-\alpha$](pro-left.west)
            \Loop[dist = 2.5cm, dir = SO, label = $1-\beta$](pro-right.east)           
    \end{tikzpicture}
    \caption{The two-state Markov chain representing the  correlation of the user's requests $X_t, t \in \mathbb{N}$.  }
\label{fig:example}
\end{figure}
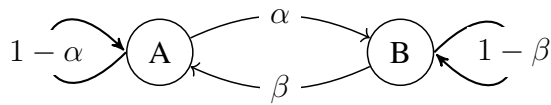

For the sake of brevity, we focus on the two time instants $t=0$ and $t=1$, and assume that privacy is  ON at $t=0$ and is switched to OFF at $t=1$. This means that the user would like to hide whether he was watching a left-leaning or a right-leaning video at time $t=0$, but does not care about revealing the source of the  video he watched at $t=1$.

  The goal is to devise an ON-OFF privacy scheme that  always gives the user the video he wants, but never reveals the choice of sources when privacy is ON, \Ie $t=0$ in this case. 
  More precisely, the server observes queries at both times $t=0$ and $t=1$, \Ie $Q_0$ and $Q_1$, which should be independent of the user's interest at time $t=0$ when privacy was ON, \Ie $X_0$.
  We are interested in schemes that minimize the download cost, or equivalently maximize the download rate (the inverse of the normalized download cost). 

%

At  $t=0$, the problem is simple. The user  achieves privacy by downloading both videos. We say that the user's  query at $t=0$ is $Q_{0}=AB$. Therefore, the download rate at $t=0$ is $R_0=1/2$.

At  $t=1$, the privacy is OFF. Now, the user must be careful not to directly declare his request, because this may reveal information about his  request at $t=0$ which is to remain private. The user can again download both videos, \Ie $Q_{1}=AB$,  and achieve privacy with a rate $R_1=1/2$. 

Our key result is that the user can achieve a better expected rate at $t=1$, without compromising privacy,  by 
\begin{itemize}
\item choosing randomly between downloading  $A$ ($Q_{1}=A$) or both $A$ and $B$ ($Q_{1}=AB$) if he wants $X_1=A$,
\item choosing randomly between downloading  $B$ ($Q_{1}=B$) or both $A$ and $B$ ($Q_{1}=AB$) if he wants  $X_1=B$. 
\end{itemize}

%
%
%
%
%
\begin{table}[t]
\normalsize
\centering
	\begin{tabular}{ c c |c c c }
	$X_{0}$ & $X_1$& $Q_1=A$ &$Q_1=B$& $Q_1=AB$\\
	 \toprule
	 $A$ & $A$ &$0.25$&$0$&$0.75$\\
	 $A$ & $B$& $0$ &$1$ &$0$\\
	 $B$ & $A$&$1$&$0$&$0$\\
	 $B$ & $B$&$0$&$0.25$&$0.75$\\

	\end{tabular}
	\caption{An example of our ON-OFF privacy scheme for $\alpha=
\beta=0.2$. The query $Q_{1}$ at $t=1$ is a probabilistic function of $X_0$ and $X_1$, the requests at $t=0$ and $t=1$ respectively. The entries of the table represent the probabilities $p(Q_{1}\mid X_{0},X_1)$. $Q_1=AB$ means that the user downloads the videos from both sources $A$ and $B$. }
	\label{table:example}
	\vspace{-10pt}
\end{table}

This random choice must also depend on the request $X_0$ at $t=0$.  The different probabilities defining the    scheme are given  in Table~\ref{table:example} and will be justified later when we explain the general scheme. For now, one can check that these probabilities lead to 
$$\Pr(Q_1 =q)=\Pr(Q_1 =q\mid X_0 =x_0),$$
 for  any $q \in \{A, B, AB\}$ and any $x_0 \in \{A, B\}$. Thus, $X_0$ and $Q_1$ are independent and  the proposed scheme in  Table~\ref{table:example} achieves perfect privacy for the request at $t=0$. 
 Moreover, the scheme ensures that the user  always obtains  the video he is requesting.





For $t=1$, the rate  $R_1=1/(2-\alpha-\beta)=0.625$, which is strictly greater than $0.5$, the rate of querying  both files. We later show that this rate is actually optimal. In fact, the values in Table \ref{table:example} were carefully chosen to achieve the privacy at the highest download rate. Any other choice of the probabilities $p(Q_{1}\mid X_{0},X_1)$ would either violate privacy or lose the optimality of the rate.


\subsection{Related Work}
The ON-OFF privacy problem for $N=2$ sources was first introduced in \cite{Naim_2019}. The  similar setting was later considered in \cite{Ye_2019,TIFS} with a more stringent requirement that the privacy of both past and future requests are preserved. The concept of ON-OFF privacy was also applied to preserve privacy of sensitive genotypes in genomics in \cite{Genotype_2020}.

The special case of the ON-OFF privacy problem in which privacy is always ON, and the user's requests are independent,  reduces to the  information-theoretic private information retrieval (PIR) problem  on a single server. In this case, the best thing the user can do is download everything \cite{Chor_1995}. Except in the case when the user can use side information, which was recently studied in \cite{Kadhe_2017}. Recently, there has been significant research activity on determining the maximum download rate of PIR with  multiple servers  (\Eg \cite{Shah_2014, Sun_2017, Tajeddine_2016, Freij-Hollanti_2017, Banawan_2018}). However,  the model there requires multiple servers and,  in the parlance of this paper,  privacy is assumed to be  always ON.

A related problem that considers privacy with correlation, namely location privacy, was studied in \cite{Shokri_2011,Shokri_2017,trace_privacy,Gunduz_2021,k-anonym_2010,Tandon}, where the correlation is usually modeled by a Markov chain and the privacy notions include $k$-anonymity \cite{k-anonym_2010}, (extended) differential privacy \cite{trace_privacy}, and distortion privacy  \cite{Shokri_2017}.
The works of  \cite{Gunduz_2021,Tandon} recently studied the information-theoretic privacy measure in location-privacy protection mechanisms, and their privacy metric was defined by the mutual information between the released data and the true traces. In this paper's language, it can be viewed as the case when privacy is always ON. 
However, in this paper, we want to prevent the adversary from inferring a selective part of the requests specified by an ON or OFF privacy status, and the simple time-sharing (switching between a private and a non-private scheme according to the privacy status) approach is not permissible due to the correlation.



\subsection{Contributions}
In this paper, we introduce a model to capture the ON-OFF privacy problem when the user is downloading data from online sources. We consider the setup in which there are $N$ information sources each generating a new message at each time $t$. At each time $t$, the user randomly chooses one of the sources and requests its latest generated message. 

The privacy constraint is information theoretic: the user wants to leak zero information about the identity of the sources in which he is interested in at each time $t$ when the privacy is ON. The main challenge stems from the fact that the user's requests are not independent. As in the previous example, we model the dependence between these requests by an $N$-state Markov chain. The goal is to design an ON-OFF privacy scheme with the maximum download rate that satisfies the user's request and guarantees the privacy of the requests made when privacy is ON.

 Our main contribution is to generalize the study of ON-OFF privacy in \cite{Naim_2019}, which focused on $N=2$ sources and  privacy being switched from ON to OFF once, to any number $N$ of sources and any ON-OFF privacy pattern. We give general  outer  and inner bounds on the 
 download rate in Theorems~\ref{theorem:outer} and \ref{theorem:inner},  respectively. We also devise an efficient algorithm to construct an ON-OFF privacy scheme achieving the inner bound.  
 We recover the optimality of the achievable scheme for $N=2$, which was proven in \cite{Naim_2019}. For $N > 2$, finding tighter outer bounds and efficient constructions of ON-OFF privacy schemes that would achieve them remains an open question.

 
The rest of the paper is organized as follows. In Section~\ref{sec:formulation}, we describe the formulation of the ON-OFF privacy problem. We present our main results in Section~\ref{sec:main}. The proof of the converse and achievability will be given in Section~\ref{section:outer} and \ref{section:inner}, respectively. A computational perspective will be discussed in Section~\ref{sec:LP}, and the optimality for $N=2$ sources will be discussed in Section~\ref{sec:proof_for_2}. 



\section{Problem Formulation and notation}
\label{sec:formulation}

\subsection{Setting}
\label{sec:setting}
A single server stores $N$ information sources  $\{\sW_i: i\in \cN\}$, where 
 $\cN:=\{1,2,\ldots,N\}$. The system is time-varying, and the time index $t$ is assumed to be discrete throughout this paper, \Ie $t \in \mathbb{N}$.
Without loss of generality, we assume that each source 
$\sW_i$ 
generates a message $W_{i,t}$ consisting of $L$ symbols at each time $t$, independently and identically according to the uniform distribution over $\{0,1\}^{L}$. Such that $\{W_{i,t}: i \in \cN, t \in \mathbb{N}\}$ are mutually independent, \Ie
\begin{equation}
 	H\left(W_{i,t}: i \in \cN, t \in \mathbb{N} \right) = \sum_{i,t} H\left(W_{i,t}\right),
 \end{equation}
and
 \begin{equation}
	H\left(W_{i,t}\right) = L \quad \forall i \in \cN, t\in \mathbb{N}.
\end{equation}

At each time $t$, the user is interested in retrieving the latest message generated by a desired source, \Ie one of the messages from $\{W_{i,t}: i \in \cN \}$. In particular, let $X_t$ be the source of interest at time $t$, which takes values in $\mathcal{N}$. In the sequel, we will call $X_t$ the \emph{user's request} at time $t$. 
Since the user is always interested in the latest message generated at time $t$, we slightly abuse the notation by dropping $t$ from $W_{i,t}$ when the time index $t$ is clear in the context, \Ie $W_{i,t}$ will be written as $W_i$ and we may write the retrieved message as $W_{X_t}$.

As mentioned previously, we are particularly interested in the case where the requests $X_t$, for $t \in \mathbb{N}$, form a time-invariant Markov chain, \Ie $\{X_t: t \in \mathbb{N}\}$ is generated by a Markov source $\sX$. The transition matrix $P$ of the Markov chain is known by both the server and the user, and the transition probability from state $i$ to state $j$ is denoted by $P_{i,j}$. We also denote the initial probability distribution of the Markov chain by $\pi_{0}$.

The user may or may not wish to hide the identity of his source of interest at time $t$. Specifically, the \emph{privacy status} $F_t$ at time $t$ can be either ON or OFF, where $F_t$ is ON when the user wishes to keep $X_t$ private, and $F_t$ is OFF when the user is not concerned with privacy. Denote $\cF=\{\text{ON},\text{OFF}\}$. We assume that the privacy status $\{F_t:t \in \mathbb{N}\}$ is generated by some information source $\sF$ that is independent of the user's requests  $\{X_t:t \in \mathbb{N}\}$. We also assume that at time $t$, $\{F_i:i \leq t\}$ is known by both the server\footnote{
It is worth noting that in our formulation  we are not interested in hiding the privacy status from the server.} and the user, for all $t\in\mathbb{N}$. 
For the ease of notation, we assume that $F_0=\text{ON}$.

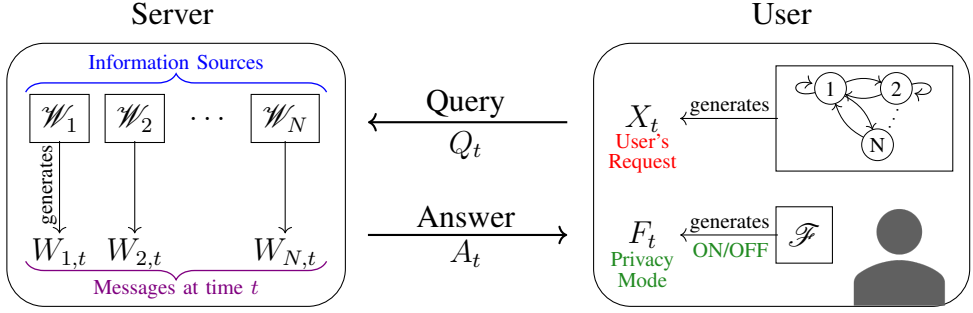
\begin{figure}[t]
    \centering
        \input{setting_figure.tex}
    \caption{Setting at time $t$ as described in Section~\ref{sec:setting}. The server stores messages $W_{1,t},\dots,W_{N,t}$ generated by information sources $\sW_1,\dots,\sW_N$, respectively. The user sends a query $Q_t$, which may be a function of all previously generated requests $\{X_i: i \leq t\}$ and privacy status $\{F_i: i \leq t\}$.
    Finally, the server replies with the answer $A_t$, which is a function of $W_{1,t},\dots,W_{N,t}$.}
    \label{fig:setting}
\end{figure}

As discussed in Section~\ref{sec:introduction}, if the user downloads the desired message at time $t$ when the privacy is OFF, the privacy in the past may be compromised. 
To ensure privacy, the user is allowed to generate unlimited local randomness and we are not interested in the amount of randomness used in this paper. The local randomness $S_t$ for $t \in \mathbb{N}$ are assumed to take values in a common alphabet $\cS$.

In this paper, we only consider a \emph{causal} system. Specifically, at time $t$, the user may utilize the \emph{causal} information, \Ie all the previous and current requests $\{X_i: i \leq t \}$, previous and current privacy status $\{F_i: i \leq t \}$, and the previously generated randomness $\{S_i: i < t \}$,
to construct a query $Q_t$, and sends to the server. In other words, the randomness $S_t$ may be generated  according to $\{X_i: i \leq t \}$, $\{F_i: i \leq t \}$ and $\{S_i: i < t \}$, \Ie 
\begin{equation}
\label{eq:formulation-St}
    S_t \sim p_{X_{[t]},F_{[t]},S_{[t-1]}},
\end{equation}
where $[t]:=\{0,1,\ldots,t\}$ and $X_{[t]}:=\{X_i: i=0,1,\ldots,t\}$. Note that \eqref{eq:formulation-St}  encompasses the case in which the current query also depends on the previous queries, since they are also functions of $\{X_i: i \leq t \}$, $\{F_i: i \leq t \}$ and $\{S_i: i < t \}$.

Upon receiving the query $Q_t$, the server 
responds to the request by producing the answer $A_t$ consisting of $\ell\left(Q_t\right)$ symbols, where $A_t$ is a function of $Q_t$ and messages $\left\{W_{i,t}:i=1,\ldots,N\right\}$, and the length of $A_t$ is a function of the query $Q_t$ received.
Thus, the average length of the answer $A_t$ is given by
\begin{equation}
  \ell_{t} = \mathbb{E}_{Q_{t}} [\ell\left(Q_t\right)].
\end{equation}
It is worth noting that $Q_t$ should be dependent of the initial distribution $\pi_0$ of the Markov chain. However, since the discussion in the sequel holds for any $\pi_0$, we drop it here for ease of notation. We can see that $\ell_{t}$ is well defined for any $\pi_0$ because $\ell\left(Q_t\right)$ is trivially bounded by $N\,L$, that is downloading all $N$ messages.

\subsection{Encoding and Decoding Functions}
\begin{definition}
	An  $(N, \sX, \sF)$ causal ON-OFF privacy system consists of the following encoding and decoding functions:
\begin{itemize}
	\item Query encoding function: 
	\[\rho_t: \cN^{t} \times \cF^{t}  \times \cS^{t}  \rightarrow \cQ, \ \  t =0,1,2,\ldots,\]
	where $\rho_t$ maps all previous (including current) requests and privacy status, together with the local randomness, to the query at time $t$, \Ie $Q_t=\rho_t\left(X_{[t]},F_{[t]},S_{[t]} \right)$.
	\item Answer length function:
	\[\ell:\cQ \rightarrow \{0,1,\ldots,NL\}, \]
	\Ie the length of the answer at time $t$ is a deterministic function of the current query, which is independent of a particular message and not time-varying over time $t$. 
	\item Answer encoding function:
	\[\phi_t: \cQ \times \{0,1\}^{NL} \rightarrow \{0,1\}^{\ell\left(\cQ\right)}, \ \  t =0,1,2,\ldots, \]
	where $\phi_t$ maps the current query and $N$ latest messages to the answer of length $\ell(Q_t)$, \Ie $A_t=\phi_t\left(Q_t,W_{1,t},\ldots,W_{N,t}\right)$.
	\item Message decoding function:
	\[\psi_t: \{0,1\}^{\ell\left(\cQ\right)} \times \cN \times \cS  \rightarrow \{0,1\}^{L}, \ \  t =0,1,2,\ldots,\]
	where $\psi_t$ maps the received answer
	to the desired message, \Ie $\hat{W}_{X_t}=\psi_t\left(A_t,X_t,S_t\right)$.
\end{itemize}
\end{definition}
We would like to emphasize two points about the setup of the model. First, for any given causal privacy status $\{F_i:i \leq t\}$ at time $t$, the query $Q_t$ may be treated as a stochastic function of all causal requests $\{X_i:i \leq t\}$ and previous queries $\{Q_i:i < t\}$. Since we are not interested in the randomness $\{S_i:i \leq t\}$ consumed, we may not write the local randomness explicitly in the sequel. 
Second, since messages $\left\{W_{i,t}: i \in \cN\right\}$ are independent over time, at time $t$, the answer $A_t$ only depends on the latest messages $W_{1,t},\ldots,W_{N,t}$ (a given $t$). Similarly, the current query $Q_t$ is independent of previous answers $\{A_i:i < t\}$ as well.

\subsection{Privacy and Decodability}
These functions need to satisfy the decodability and the privacy constraints, defined as follows.
\begin{enumerate}
\item Decodability: For any time $t$, the user should be able to recover the desired message from the answer with zero-error probability, \Ie
	\begin{equation}
	\label{eq:decode}
		\Pb{\hat{W}_{X_t} \neq W_{X_t}} = 0.
	\end{equation}
\item Privacy: For any time $t$, given all past queries received by the server, the query $Q_t$ should not reveal any information about all the past or present requests when the privacy is ON, that is
	\begin{equation}
	\label{eq:privacy-1}
		I \left(X_{\cB_t};Q_t|Q_{[t-1]}\right) = 0, \quad \forall t \in \mathbb{N},
	\end{equation}
	where $\cB_t: = \{i: i \leq t,F_i=\text{ON}\}$. For notational simplicity, $F_0$ is assumed to be $\text{ON}$ throughout this paper, and hence $\cB_t$ is always not empty. 


\end{enumerate}
The conditioning in the privacy formulation in \eqref{eq:privacy-1}
serves to       ensure causality in the proposed achievable schemes. Barring this conditioning,    privacy could be alternatively defined by 
\begin{equation}
\label{eq:privacy-joint}
    I \left(X_{\cB_t};Q_{[t]}\right) = 0, \quad \forall t \in \mathbb{N}.
\end{equation}
However, this alternative definition implies that at any point $i<t$, the user has to know and protect  future requests $\left\{X_j: j=i+1,\ldots,t, F_j=\text{ON} \right\}$, since \eqref{eq:privacy-joint} implies that
\[I \left(X_{\cB_t\backslash[i]};Q_{i} \right) = 0, \]
which contradicts  the causality of the system.

Given the definition of the privacy, we introduce the following proposition, which is a direct but useful consequence of the Markov assumption of the requests and the privacy definition and whose proof can be found in Appendix~\ref{appendix:proposition-achieve-markov}. 
\begin{proposition}
\label{proposition:achieve-markov}
If $X_{\tau}$ is independent of $Q_t$ conditioning on $Q_{[t-1]}$, then $X_{\cB_t}$ is independent of $Q_t$ conditioning on $Q_{[t-1]}$.
\end{proposition}

By convention, at time $t$, the tuple $\ell_t$ is said to be achievable if there exists a code satisfying the decodability and the privacy constraint such that the average answer length is $\ell_t$. The efficiency of the code can be measured by the download rate $R_t=\frac{L}{\ell_t}$, and hence we define the achievable region as follows.
\begin{definition}
The rate tuple $\left(R_t:t \in \mathbb{N} \right)$ is achievable if there exists a code with  message length $L$ and average  download cost $\ell_t$ such that $R_t \leq L/\ell_t$ for all $t \in \mathbb{N}$. 
\end{definition}

We are interested in characterizing the achievable region $\left(R_t:t \in \mathbb{N} \right)$. In particular, the focus of this paper is the characterization of $R_t$ for each $t \in \mathbb{N}$.

\subsection{Notation}

We introduce some necessary notation which will be used in later sections.
Let $\tau(t)$ be the last time privacy was ON, \Ie 
\begin{equation}
\label{eq:tau-def}
	\tau(t):= \max\{i: i \leq t, F_i = \text{ON}\}=\max \,\cB_t.
\end{equation}
The time index $t$ will be clear in the  context in the following sections, so we may drop $t$ from the notation and write $\tau$ instead of $\tau(t)$ for simplicity. 
It is worth noting that $\tau(t)$ is well-defined because of the assumption that $F_0 = \text{ON}$.

For any given $x \in \cN$ and $q_{[t-1]}$, suppose that we have the following ordering of the likelihood probabilities
\begin{equation}
\label{eq:prob_order}
\begin{aligned}
& \pb{X_t=x|X_{\tau} = x_{\tau}^{(x,1)},Q_{[t-1]}=q_{[t-1]}} \leq \pb{X_t=x|X_{\tau} = x_{\tau}^{(x,2)},Q_{[t-1]}=q_{[t-1]}} \\
& ~~~~~~ \leq \cdots \leq \pb{X_t=x|X_{\tau} = x_{\tau}^{(x,N)},Q_{[t-1]}=q_{[t-1]}},
\end{aligned}
\end{equation}
where $x_{\tau}^{(x,i)}$ for $i=1,\ldots,N$ are distinct elements in $\cN$. Then, for $i=1,\ldots,N$, let
\begin{equation}
    \label{eq:def_lambda}
	\lambda_i\left(t,q_{[t-1]}\right) = \sum_{x \in \cN} \pb{X_t=x|X_{\tau} = x_{\tau}^{(x,i)},Q_{[t-1]} = q_{[t-1]}},
\end{equation}
and
\begin{equation}
    \label{eq:def_theta}
	\theta_i\left(t,q_{[t-1]}\right) = \min \left\{ 1,\lambda_i\left(t,q_{[t-1]}\right)\right\} - \min \left\{ 1,\lambda_{i-1}\left(t,q_{[t-1]}\right) \right\}, 
\end{equation}
where $\lambda_{0}\left(t,q_{[t-1]}\right)$ is assumed to be $0$. For notational simplicity, we will also write $\lambda_i\left(t,q_{[t-1]}\right)$ by $\lambda_i\left(q_{[t-1]}\right)$ and $\theta_i\left(t,q_{[t-1]}\right)$ by $\theta_i\left(q_{[t-1]}\right)$  when the time index $t$ is clear in the  context.

Moreover, we will use $\sP\left(\cN\right)$ to denote the power set of $\cN$, and $\mathbb{E}[X]$ to denote the expected value of a random variable $X$. 
We summarize some definitions and nomenclature in Table~\ref{table:nomenclature}.


\begin{table}[t]
\centering
\small

\begin{tabular}{ l|l}

Symbol & Definition\\
\toprule
$N$ & number of sources \\[0.5ex]

$\cN$ & $\{1,2,\ldots,N\}$  \\[0.5ex]

$[t]$ & $\{0,1,\ldots,t\}$ for any $t \in \mathbb{N}$  \\[0.5ex]
 
 
$\sP\left(\cN\right)$ & power set of $\cN$\\[0.5ex]


$W_{i,t}$ & message generated by $i$-th source at time $t$, where $i=1,\ldots,N$ and $t=0,1,\ldots$ \\[0.5ex]

$X_t$ & user's request at time $t$ ($X_t \in \cN $)\\[0.5ex]


$F_t$ & privacy status at time $t$, \Ie $F_t \in \{\text{ON}, \text{OFF}\}$\\[0.5ex]



 $Q_t$ & query sent by the user to the server at time $t$\\[0.5ex]
 
 $A_t$ & answer sent by the server to the user at time $t$\\[0.5ex]
 
 
$\cB_t$ & all the times privacy was ON , \Ie $\cB_t = \{i: i \leq t,F_i=\text{ON}\}$  \\[0.5ex]
 
$\tau(t)$ & last time privacy was ON, \Ie  $\tau(t) = \max \cB_t$ \\[0.5ex]
 
 $\ell_{t}$ & average length of the answer $A_t$\\[0.5ex]
 
 $R_t$ & download rate at time $t$\\[0.5ex]

$\lambda_i\left(q_{[t-1]}\right)$
&  the summation of $i$-th minimal likelihood probabilities (of $x_{\tau}$) provided the observation $x_t$ for given $q_{[t-1]}$
\\[0.5ex]


\toprule

\end{tabular}
\caption{Nomenclature and definitions}
\label{table:nomenclature}
\end{table}

\section{Main results}
\label{sec:main}
In this section, we present the main results of this paper, \Ie inner and outer bounds for the achievable region $\left(R_t:t \in \mathbb{N} \right)$.

The following theorem gives an outer bound on the achievable rate, and the proof can be found in Section~\ref{section:outer}.
\begin{theorem}(Outer bound 1)
	\label{theorem:outer}
	The rate tuple $\left(R_t:t \in \mathbb{N} \right)$ must satisfy 
	\begin{equation}
	\label{eq:theorem-outer}
		\frac{1}{R_t} \geq \sum_{q_{[t-1]}} \pb{q_{[t-1]}} \sum_{x_t} \max_{x_{\tau}} \pb{x_t|x_{\tau},q_{[t-1]}},
	\end{equation}
	where $\tau = \max\{i: i \leq t, F_i = \text{ON}\}$. 
\end{theorem}

It is worth noting that the right-hand side of \eqref{eq:theorem-outer} encompasses the previous queries, where the optimal previous queries maximizing the download rate for the current time instance are implicit, and hence the bound in \eqref{eq:theorem-outer} is generally hard to compute. Nevertheless, we can use the bound in \eqref{eq:theorem-outer} to derive the following corollary, which only involves the transition probabilities of the Markov chain and not the previous queries.

\begin{corollary}(Outer Bound 2)
	\label{corollary:outer}
	The rate tuple $\left(R_t:t \in \mathbb{N} \right)$ must satisfy 
	\begin{equation}
	\label{eq:corollary}
		\frac{1}{R_t} \geq \sum_{x_t} \max_{x_{\tau}} \pb{x_t|x_{\tau}},
	\end{equation}
	 where $\tau = \max\{i: i \leq t, F_i = \text{ON}\}$. 
\end{corollary}
\begin{proof}
See Appendix~\ref{appendix:corollary-outer}.
\end{proof}


The following theorem gives an inner bound on the rate, and the detailed description of the achievable scheme will be discussed in Section \ref{sec:proof_inner}. 
\begin{theorem}(Inner bound)
	\label{theorem:inner}
	The rate tuple $\left(R_t:t \in \mathbb{N} \right)$ is achievable if 
	\begin{equation}
	\label{eq:theorem-inner}
	 	\frac{1}{R_t} \geq \sum_{q_{[t-1]}} \pb{q_{[t-1]}} \sum_{i=1}^{N} i \, \theta_i(q_{[t-1]}).
	 \end{equation} 
\end{theorem}

We give the following example to illustrate the outer and inner bounds described in Theorem \ref{theorem:outer} and Theorem \ref{theorem:inner}, respectively.




\begin{example}
\label{ex:symmetric_markov}
Consider a symmetric Markov chain with transition matrix $P$ given by
\begin{equation}
    P_{i,j} =
    \begin{cases}
        \alpha, & \quad \text{if} \quad i=j,\\
        \frac{1-\alpha}{N-1}, & \quad \text{if} \quad i\neq j,
    \end{cases}
\end{equation}
where $0 \leq \alpha \leq 1$ and $P_{i,j}$ denotes 
the transition probability from state $i$ to state $j$. 

Suppose we are given $\tau=0$, \Ie privacy was ON at $t=0$, and privacy is OFF at $t=1$.   
Following a direct application of \eqref{eq:theorem-outer} and \eqref{eq:theorem-inner} for $t=1$, we have two regimes: $\alpha<\frac{1}{N}$ and $\alpha\geq\frac{1}{N}$. This is because the ordering of probabilities (c.f.\eqref{eq:prob_order}) changes at $\alpha = \frac{1}{N}$.

For $\alpha \geq \frac{1}{N}$, the bounds 
\eqref{eq:theorem-outer} and \eqref{eq:theorem-inner} match, \Ie the rate at $t=1$ is achievable if and only if
\begin{equation}
\label{eq:exp-1}
  \frac{1}{R_1}\geq N\alpha.
\end{equation}
As for $\alpha<\frac{1}{N}$, we have that 
\begin{equation}
\label{eq:exp-2}
    \frac{1}{R_1^O}:=\frac{N(1-\alpha)}{N-1}\leq \frac{1}{R_1} \leq 2-N\alpha:=\frac{1}{R_1^I}.
\end{equation}

We illustrate \eqref{eq:exp-1} and \eqref{eq:exp-2} for a three state symmetric Markov chain, i.e., $N=3$, with more details in  
Figure~\ref{fig:example_3states_sym}.
\end{example}
\begin{figure*}[t]
     \centering
     \begin{subfigure}[l]{0.45\textwidth}
         \centering
          \scalebox{0.9}{                  \input{symmetric_markov}
         }
         \caption{\centering A Symmetric Markov Chain.}
         \label{fig:example_3states_chain}
     \end{subfigure}
     \begin{subfigure}[l]{0.45\textwidth}
         \centering
         \input{symmetric_plot}
         
         \caption{\centering $R_t^I$ and $R_t^O$ as a function of $\alpha$.}
         \label{fig:example_3states_p}
     \end{subfigure}
     
        \caption{ In Figure~\ref{fig:example_3states_chain}, we graphically represent the 3-state symmetric Markov chain used in Example~\ref{ex:symmetric_markov}, where $0 \leq \alpha\leq 1$. In Figure~\ref{fig:example_3states_p}, we plot the achievable rate $R_t^I$ and the upper bound $R_t^O$ (c.f.\eqref{eq:exp-2}),
         as a function of $\alpha$, when $\tau=0$ and $t=1$.
         }
        \label{fig:example_3states_sym}
\end{figure*}
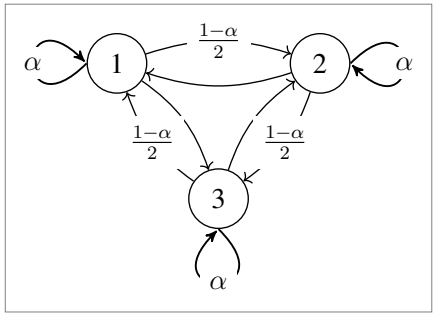
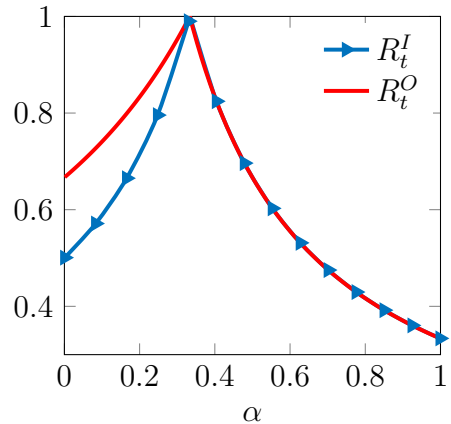

The special case when there are $N=2$ information sources was studied in \cite{Naim_2019}. For $N=2$, the outer bound \eqref{eq:theorem-outer} and inner bound \eqref{eq:theorem-inner}, presented above, match. Therefore, the proposed scheme achieves the optimal rate for $N=2$. We restate this result in Theorem~\ref{theorem:two}, where the Markov chain has two states and is defined by the probability transition matrix
\begin{equation}
 P =
\begin{bmatrix}
       1- \alpha &  \alpha         \\
       \beta  & 1- \beta 
\end{bmatrix},
\end{equation}   
such that $0 \leq \alpha, \beta \leq 1$.

\begin{theorem}(Optimality for $N=2$)
\label{theorem:two}
	For $N=2$, the rate tuple $\left(R_t:t \in \mathbb{N} \right)$ is achievable if and only if  
  \begin{equation}
  \label{eq:theorem-two}
      \frac{1}{R_t} \geq 1+ |1- \alpha -\beta|^{t-\tau},
  \end{equation}
  where $\tau = \max\{i: i \leq t, F_i = \text{ON}\}$. 
\end{theorem}
\noindent

\begin{table*}[b]	
    \centering
    \begin{tabular}{l|ccc||ccc||ccc}
	\diagbox[dir=SE,width=58pt,height=22pt]{\scriptsize$X_{\tau},X_{t}$}{\scriptsize$Q_{t}$} & $A$ & $B$ & $AB$ &$A$ & $B$ & $AB$&$A$ & $B$ & $AB$\\[\smallskipamount]
	\toprule
	$A,A$ &     $\frac{\beta}{1-\alpha}$            & 0     &  $\frac{1-\alpha - \beta}{1-\alpha}$&$\frac{1-\alpha}{\beta}$            & 0     &  $\frac{\alpha + \beta-1}{\beta}$&1&0&0\\[\smallskipamount]
	
	$A,B$ & 0  &      $1$              &     0 &$ 0$  &      $1$              &     $0$    &0&    $\frac{1-\beta}{\alpha}$              &     $\frac{\alpha+\beta-1}{\alpha}$    \\[\smallskipamount]
	
	$B,A$ &       $1$          & 0     &  0    &$1$  &  0&0  &  $ \frac{1-\alpha}{\beta}$    &0        &  $\frac{\alpha + \beta-1}{\beta}$ \\[\smallskipamount]
	
	$B,B$ & 0  &  $ \frac{\alpha}{1-\beta}$                    &  $\frac{1-\alpha - \beta}{1-\beta}$&0  &  $ \frac{1-\beta}{\alpha}$                    &  $\frac{\alpha + \beta-1}{\alpha}$&0  &  1  & 0\\[\smallskipamount]
	\toprule
	\multicolumn{1}{c}{} & \multicolumn{3}{c}{(a) $\alpha+\beta< 1$} & \multicolumn{3}{c}{(b) $\alpha+\beta> 1$ and $t$ is even} & \multicolumn{3}{c}{(c) $\alpha+\beta> 1$ and $t$ is odd} 
	\end{tabular}
	\caption{The optimal ON-OFF privacy scheme that achieves the bound in \eqref{eq:theorem-two} for $N=2$. The query $Q_t$ is probabilistic and depends on the current request $X_t$, the previous query $Q_{t-1}$ and the last private request $X_{\tau}$.  The scheme consists of the following two cases: (i) if $Q_{t-1} = \{1,2\}$, \Ie the previous query was for two messages, then the current query $Q_t=X_t$; (ii) if $Q_{t-1} \neq \{1,2\}$, \Ie the previous query was for one message, then the current query $Q_{t}$ is chosen based on the probabilities $\pb{q_t|x_{\tau},x_t,q_{t-1}}$ given in this table. For (a) $\alpha+\beta<1$, (b) and (c) are for $\alpha+\beta>1$ where $t-\tau$ is even or odd respectively \cite{Naim_2019}.
	}
	   \label{table:Qt}
\end{table*}

Theorem \ref{theorem:two} reflects the fact that, when the Markov chain is ergodic, the information carried by $X_t$ about $X_\tau$ is decreasing exponentially as $t-\tau$ grows,
so the user can eventually directly ask for the desired message at time $t$ without being concerned about leaking information about $X_{\tau}$.  
Table \ref{table:Qt} gives an explicit scheme that achieves the rate in \eqref{eq:theorem-two}. The details of this construction will be further discussed in Section \ref{sec:two-achievability}. Figure~\ref{fig:th_plot} shows the rate $R_t$ as a function of time for different values of $\alpha+\beta$. As $\alpha+\beta$ approaches $1$, the correlation between the request decreases leading to an increase in the rate.

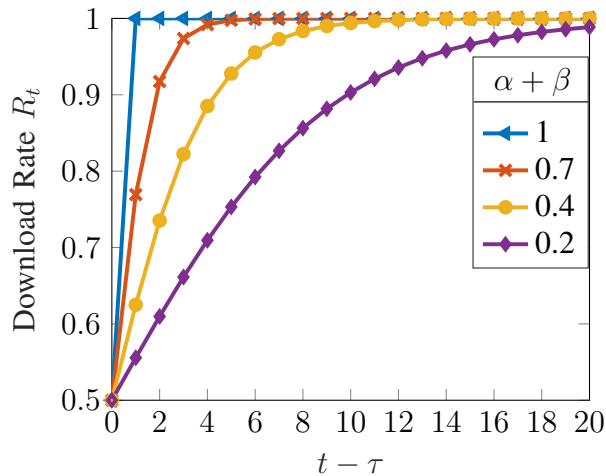
\begin{figure}[t]
    \centering
    \input{plot.tex}
    \caption{The maximum rate $R_t$, as given in Theorem \ref{theorem:two}, as a function of $t-\tau$ for different values of $\alpha+\beta$. As $\alpha+\beta$ approaches $1$, the correlation between the requests decreases leading to an increase in the rate. For $\alpha+\beta=1$, the requests are independent. In this case, when privacy is ON at time $t$, which means $t-\tau=0$, the user has to download both messages, \Ie $R_t=1/2$. When privacy is OFF at time $t$
, which means $t-\tau>0$, the user only downloads the desired message, \Ie $R_t=1$.}
    \label{fig:th_plot}
    \vspace{-10pt}
\end{figure}

\section{Proof of the Outer Bound in Theorem \ref{theorem:outer}}
\label{section:outer}

Recall that the inverse of the rate is expressed as
\begin{equation}
\label{eq:rate-lower}
	\frac{1}{R_t} = \frac{\ell_t}{L} = \frac{1}{L} \mathbb{E}\left[\ell(Q_t) \right]. 
\end{equation}
Hence, to obtain an upper bound on the rate $R_t$ (a lower bound on $1/R_t$), we will derive a lower bound on the average downloading cost $\mathbb{E}[\ell(Q_t)]$ under the privacy and the decodability constraints. 

First, we define an auxiliary random variable $Y_t$ taking values in $\sP\left(\cN\right)$ based on the decodability of the subset of messages.   
Specifically, let $Y_t$ be a function of $Q_t$ such that $Y_t = \cD$ for $\cD \in \sP\left(\cN\right)$ if the answer $A_t$ can decode the messages $W_{\cD}$ but not any message $W_{i}$ for $i \in \cN\backslash \cD$. 
Roughly speaking, $Y_t$ represents the capability of decoding  messages from the query $Q_t$.
Note that since the query $Q_t$ and messages $W_{\cN}$ are independent, the decodability of any message is known by the server only through $Q_t$, that is, $Y_t$ is a function of $Q_t$. 
In this way, the alphabet $\cQ$ of the query is partitioned into $2^N$ classes based on the decodability of the subset of the messages. Clearly, from the definition of $Y_t$, we have
\begin{equation}
	\ell(Q_t) \geq |Y_t|\,L,
\end{equation}
and hence \eqref{eq:rate-lower} can be written as
\begin{equation}
\label{eq:rate-lower2}
	\frac{1}{R_t} \geq \mathbb{E}\left[|Y_t|\right].
\end{equation}

Thus, it remains for us to give a lower bound on $\mathbb{E}\left[|Y_t|\right]$ under the privacy and the decodability constraints. 

Now, we start to interpret the privacy and the decodability constraints.
By the definition of $Y_t$, the decodability can be rewritten as 
\begin{equation}
\label{eq:converse-decode-prob}
	\pb{x_t,y_t}=0, \forall x_t \notin y_t.
\end{equation}
Recall the privacy constraint that we require is
\begin{equation*}
  I \left(X_{\cB_t};Q_t|Q_{[t-1]}\right) = 0.
\end{equation*} 
Since 
\begin{align*}
	 I \left(X_{\cB_t};Q_t|Q_{[t-1]}\right) 
	& \geq I \left(X_{\tau};Q_t|Q_{[t-1]}\right) \geq I \left(X_{\tau};Y_t|Q_{[t-1]}\right),
\end{align*}
we can relax the privacy constraint by 
\begin{equation}
\label{eq:converse-privacy-ineq}
 I \left(X_{\tau};Y_t|Q_{[t-1]}\right) = 0.
\end{equation}

Therefore, to obtain an upper bound on the rate $R_t$ (a lower bound on $1/R_t$), it remains for us to give a lower bound on $\mathbb{E}\left[|Y_t|\right]$ such that \eqref{eq:converse-decode-prob} and \eqref{eq:converse-privacy-ineq} are satisfied, which relies on the following lemma. The proof of the lemma can be found in Appendix \ref{Appendix-lemma-lower}.
\begin{lemma}
\label{lemma:lower}
For any random variables $U$, $X$ and $Y$, taking values in the alphabet $\cN$, $\cN$ and $\sP\left(\cN\right)$ respectively, if $Y$ is independent of $U$, and $p(x,y)=0$ for $x \notin y$, then 
	\begin{equation}
	\label{eq:lemma-lower}
		\mathbb{E}\left[|Y|\right] \geq 
		 \sum_{x \in \cN} \max_{u \in \cN} p\left(x|u\right). 
	\end{equation}
\end{lemma}





For any given $q_{[t-1]}$, we can see that Lemma~\ref{lemma:lower} immediately gives a lower bound on $\mathbb{E}\left[|Y_t||q_{[t-1]}\right]$, \Ie
\begin{equation}
	\mathbb{E}\left[|Y_t||q_{[t-1]}\right] \geq  \sum_{x_t} \max_{x_{\tau}} \pb{x_t|x_{\tau},q_{[t-1]}}.
\end{equation}
Thus, by summing over all $q_{[t-1]}$, we can obtain that
\begin{equation}
\label{eq:converse-cardinality-lower}
	\mathbb{E}\left[|Y_t|\right] = \sum_{q_{[t-1]}} \pb{q_{[t-1]}}  \mathbb{E}\left[|Y_t||q_{[t-1]}\right] \geq \sum_{q_{[t-1]}} \pb{q_{[t-1]}} \sum_{x_t} \max_{x_{\tau}} \pb{x_t|x_{\tau},q_{[t-1]}}.
\end{equation}

By substituting \eqref{eq:converse-cardinality-lower} in \eqref{eq:rate-lower}, we finally get
\[\frac{1}{R_t} \geq  \sum_{q_{[t-1]}} \pb{q_{[t-1]}} \sum_{x_t} \max_{x_{\tau}} \pb{x_t|x_{\tau},q_{[t-1]}},\]
which completes the proof. 


\section{Inner Bound in Theorem~\ref{theorem:inner}}
\label{section:inner}
Before we move on to describe the achievable  scheme,
we present an example for $N=3$ sources, which 
illustrates the basic idea of 
the scheme that achieves the bound in \eqref{eq:theorem-inner}.
\subsection{Example of an achievable scheme}
\label{ex:achievability_N}
Suppose the transition probabilities of the Markov chain are given by
    \begin{equation}
    \label{eq:example-tran-prob}
        P =
            \begin{bmatrix}
                   0.1 & 0.3 & 0.6 \\
                   0.5 & 0.4 & 0.1 \\
                   0.2 & 0.5 & 0.3
            \end{bmatrix},
    \end{equation}   
where $P_{i,j}=\Pb{X_t=j|X_{t-1}=i}$. 

Assume that privacy is ON at time $t=0$ and privacy is OFF at time $t=1$.
At time $t=0$, we know the user has to send the query $Q_0=\{1,2,3\}$. Our goal is to design the query $Q_1$ at $t=1$.  In particular, in this example,
the query $Q_1$ is uncoded and is a probabilistic function of the previous request $X_0$, the current request $X_1$ and the previous query $Q_0$. As such, we will show how to design the query encoding function $\pb{q_1|x_1,x_0}$\footnote{We drop $q_0$ in  $\pb{q_1|x_1,x_0, q_0}$ since $q_0=\{1,2,3\}$ is a constant.}, or equivalently $\pb{q_1,x_1|x_0}$, for all $x_0,x_1 \in  \{1,2,3\}$ and $q_1 \in \sP\left(\{1,2,3\}\right)$.
The distribution $\pb{q_1,x_1|x_0}$ is represented in Table~\ref{table-example}. Throughout this example, we will show how to fill in the values of the cells in Table~\ref{table-example}.

    \begin{table*}[t]	
        \centering
        \begin{tabular}{|c|c|c|c|c|c|c|c|c||>{\columncolor{lightblue}}c|}
             \thickhline
            \multicolumn{2}{|l|}{\diagbox[dir=SE,width=55pt,height=33pt]{\hspace{-1pt}$x_0$\hspace{15pt} $x_1$}{\vspace{-32pt}$q_1$}} 
            & $\{1\}$ & $\{2\}$ & $\{3\}$ & $\{1,2\}$ & $\{1,3\}$ & $\{2.3\}$ & $\{1,2,3\}$ &\textcolor{black}{Budget $\left(P_{i,j}\right)$}\\
            \thickhline 
            
            \multirow{3}{*}{$1$} & $1$ & $0.1$ & \cellcolor{light_gray}$0$ & \cellcolor{light_gray}$0$ & $0$ & $0$ & \cellcolor{light_gray}$0$ & $0$ &\textcolor{black}{$0.1$}\\
            \cline{2-10} 
            & $2$ & \cellcolor{light_gray}$0$ & $0.3$ & \cellcolor{light_gray}$0$ & $0$ & \cellcolor{light_gray}$0$ & $0$ & $0$ &\textcolor{black}{$0.3$}\\
            \cline{2-10} 
             & $3$ & \cellcolor{light_gray}$0$ & \cellcolor{light_gray}$0$ & $0.1$ & \cellcolor{light_gray}$0$ & $0.1+0.2$ & $0.1$ & $0.1$ &\textcolor{black}{$0.6$}\\
            \thickhline 
            \multirow{3}{*}{$2$} & $1$ & $0.1$ & \cellcolor{light_gray}$0$ & $\cellcolor{light_gray}0$ & $0$ & $0.1+0.2$ & \cellcolor{light_gray}$0$ & $0.1$ &\textcolor{black}{$0.5$}\\
            \cline{2-10} 
             & $2$ & \cellcolor{light_gray}$0$ & $0.3$ & \cellcolor{light_gray}$0$ & $0$ & \cellcolor{light_gray}$0$ & $0.1$ & $0$ &\textcolor{black}{$0.4$}\\
            \cline{2-10} 
             & $3$ & \cellcolor{light_gray}$0$ & \cellcolor{light_gray}$0$ & $0.1$ & \cellcolor{light_gray}$0$ & $0$ & $0$ & $0$ &\textcolor{black}{$0.1$}\\
            \thickhline
            \multirow{3}{*}{$3$} & $1$ & $0.1$ & \cellcolor{light_gray}$0$ & \cellcolor{light_gray}$0$ & $0$ & $0.1$ & \cellcolor{light_gray}$0$ & $0$ &\textcolor{black}{$0.2$}\\
            \cline{2-10} 
             & $2$ & \cellcolor{light_gray}$0$ & $0.3$ & \cellcolor{light_gray}$0$ & $0$ & \cellcolor{light_gray}$0$ & $0.1$ & $0.1$ &\textcolor{black}{$0.5$}\\
            \cline{2-10} 
             & $3$ & \cellcolor{light_gray}$0$ & \cellcolor{light_gray}$0$ & $0.1$ & \cellcolor{light_gray}$0$ & $0.2$ & $0$ & $0$ &\textcolor{black}{$0.3$}\\
             \toprule
    	\end{tabular}
    	  \caption{The constructed distribution $\pb{q_1,x_1|x_0}$ for the given $\pb{x_1|x_0}$ in Example~\ref{ex:achievability_N}.}
    	   \label{table-example}
    \end{table*}

As requested, the query $Q_1$ should satisfy the decodability and the privacy constraints. The two constraints can be translated into the following rules for filling Table \ref{table-example}.
\begin{enumerate}
   \item Satisfying the decodability constraint is straightforward. We set $\pb{q_1,x_1|x_0}=0$ for all $x_1\notin q_1$, \Ie setting all the gray highlighted cells in Table~\ref{table-example} to zero. This guarantees that the user always receives messages containing the one he wants when the server responds to his query.
   
   \item The privacy constraint requires that $Q_1$ is independent of $X_0$, \Ie
   \[\pb{Q_1=q_1|X_0=1}=\pb{Q_1=q_1|X_0=2}=\pb{Q_1=q_1|X_0=3},\]
for all $q_1 \in \sP\left(\{1,2,3\}\right)$. By the law of total probability, this can be written as 
   \begin{equation}
       \sum_{x_1}\pb{q_1,x_1|X_0=1}=\sum_{x_1}\pb{q_1,x_1|X_0=2}=\sum_{x_1}\pb{q_1,x_1|X_0=3}.
   \end{equation}
   To translate this in Table \ref{table-example}, each column is divided into 3 blocks (pertaining to $x_0\in \{1,2,3\}$), and the sum of the cells in each block in a given column are to be equal, \Eg in column $\{1,3\}$ each block sum to $0.3$.
   
   \item Since the entries are probabilities, this requires the sum of row $j$ in a given block $i$ to be equal to $\pb{X_1=j|X_0=i}$, \Ie $P_{i,j}$ in the matrix $P$. We will refer to $P_{i,j}$ as our \textit{budget} for row~$j$ in block~$i$, it is highlighted in blue in Table~\ref{table-example}. 
\end{enumerate}

     We now introduce an ordering of probabilities, such that $$\Pb{X_1=j|X_0=x_{0}^{(j,1)}}\leq\Pb{X_1=j|X_0=x_{0}^{(j,2)}}\leq\Pb{X_1=j|X_0=x_{0}^{(j,3)}}$$ for each $j \in \{1,2,3\}$.  For example, for $X_1=1$, we observe that $P_{1,1}\leq P_{1,3} \leq P_{1,2}$, so $x_{0}^{(1,1)}=1$, $x_{0}^{(1,2)}=3$, and $x_{0}^{(1,3)}=2$. We summarize the values of the rest of the variables in the Table~\ref{table:example-precal}.

        \begin{table}[t]
	\normalsize
\centering
	\begin{tabular}{  c |c c c || c || c}

	 $x_{0}^{(j,i)}$& $j=1$ &$j=2$& $j=3$&$\lambda_i$&$\theta_i$\\
	 \toprule
	 $x_{0}^{(j,1)}$ &$1$&$1$&$2$&$0.5$&$0.5$\\
	 $x_{0}^{(j,2)}$& $3$ &$2$ &$3$&$0.9$&$0.4$\\
	 $x_{0}^{(j,3)}$&$2$&$3$&$1$&$1.6$&$0.1$\\
	 \hline

	\end{tabular}
    \caption{Useful Variables for Example \ref{ex:achievability_N}.}
	\label{table:example-precal}
\end{table}
 
It is worth noting that downloading all messages is always a feasible solution here. More precisely, setting the probability of querying three messages to be equal to the budget, i.e., $$p\left(Q_1=\{1,2,3\},X_1=x_1|X_0=x_0\right)=p\left(X_1=x_1|X_0=x_0\right)$$ for all $x_0,x_1\in\{1,2,3\}$, always satisfies rules one-three. 
Next, we present the algorithm that better fills the table and satisfies the aforementioned rules. 
The main idea is to assign values as large as possible to $Q_1$ with small cardinality, and this will ultimately lower the communication cost.
\begin{itemize}     
\item \textbf{Step 1:} 
We start with queries $q_1$ of cardinality one, \Ie $|q_1|=1$.  
We adopt a greedy-like approach, which means we try to maximize the value filled in the first three columns. Due to the second and third rules 
mentioned above, the maximum values we can choose are
    \begin{align}
    \label{eq:example_step1}
    \hspace{-50pt}
	\pb{Q_1=\{x_1\},X_1=x_1|X_0=x_0} = \min_{x_0} \pb{x_1|x_0} 
	= \pb{x_1|x_0^{(x_1,1)}}
	=\begin{cases}
			0.10, & x_1=1,\\
			0.30, & x_1=2,\\
			0.10, & x_1=3.
		\end{cases}	
	\end{align}
	
Note that in some rows the rest of the cells, \Eg row $1$ in block $1$, have to be zero, because from rule 3 we know that their budget has been consumed. 
	
        

\item 
\textbf{Step 2:} 
When $|q_1|=2$, the construction is
more complicated because each block has two cells to fill. We describe it as follows. 
\begin{itemize}
\item[$\circ$] For $X_1 = 1$, we know that $x_0^{(1,1)}=1$ and $x_0^{(1,2)}=3$. 
Since, in Step 1 \eqref{eq:example_step1}, we consumed the probability $\pb{X_1=1|X_0=x_0^{(1,1)}}$, we deduct it from the the second minimal value $\pb{X_1=1|X_0=x_0^{(1,2)}}$, and calculate 
\[\pb{X_1=1|X_0=x_0^{(1,2)}} - \pb{X_1=1|X_0=x_0^{(1,1)}} = 0.1.\]
Then, we may find some $\hat{q}$ (to be determined), such that $|\hat{q}|=2$ and $1 \in \hat{q}$ and set 
\begin{equation}
\label{eq:example-assign-2}
\begin{aligned}
& \pb{Q_1=\hat{q},X_1=x_1|X_0=x_0} \\
& ~~~~ = 
\begin{cases}
    0.1, & x_1=1 \wedge x_0 \neq x_0^{(1,1)}~ \text{or}~ x_1=\hat{q}\backslash \{1\} \wedge x_0 = x_0^{(1,1)}, \\
    0, & \text{others}.
\end{cases}
\end{aligned}
\end{equation}
Here, we have two options for $\hat{q}$, either $\{1,2\}$ or $\{1,3\}$. If $\hat{q}=\{1,2\}$, from rule 2, we know that the summation of each block must be the same. However, if we inspect first block \Ie $X_0=1$, we can find that the budget for the first two rows of the first block is zero, which means that we do not have enough budgets to assign values according to \eqref{eq:example-assign-2}.
Therefore, if we choose $\hat{q}=\{1,2\}$, then it will violate rule 2, so that $\hat{q}$ is chosen to be $\{1,3\}$, and fill in the table according to \eqref{eq:example-assign-2}.

\item[$\circ$] For $X_1=2$ the procedure is the same as we did for  $X_1=1$ and details are omitted. 

	    
	    
\item[$\circ$] 
For $X_1 = 3$, we know that $x_0^{(3,1)}=2$ and $x_0^{(3,2)}=3$. Also, we have \[\pb{X_1=3|X_0=x_0^{(3,2)}} - \pb{X_1=3|X_0=x_0^{(3,1)}} = 0.2.\]

Then, we follow the same procedure as above by determining $\hat{q}=\{1,3\}$. However, since we have assigned a value $0.1$ to the cell $\pb{Q_1=\{1,3\},X_1=3|X_0=1}$ in previous steps, 
we augment its value by $0.2$, and finally we have 
\[\pb{Q_1=\{1,3\},X_1=3|X_0=1}=0.1+0.2=0.3.\]

	\end{itemize}

    \item \textbf{Step 3:}
        When $|q_1|=3$, since this is the last column, we just need to complete the table such that the budget of all rows is fully consumed.
        
\end{itemize}

Finally, let us evaluate the achievable rate $R_1$, equivalently $1/\mathbb{E}[|Q_1|]$, achieved by the constructed $\pb{q_1,x_1|x_0}$. It is easy to see that we assign $\theta_1=\lambda_1 = 0.5$ to cells such that $|q_1|=1$, $\theta_ 2= \lambda_2 - \lambda_1 = 0.4$ to cells such that $|q_1|=2$, and  $\theta_3= 1 - \lambda_2 = 1-0.9 = 0.1$ to cells such that $|q_1|=3$ for each block, so that we have
\[\mathbb{E}[|Q_1|] = \sum_{i=1}^{3} i \, \theta_i=1.6,\]
where $\lambda_i$ and $\theta_i$ are defined in \eqref{eq:def_lambda} and \eqref{eq:def_theta} respectively. Thus, $R_1=5/8$ is achievable in this example. 
One may notice that the outer bound in Corollary~\ref{corollary:outer} gives 
\[\frac{1}{R_1} \geq \sum_{x_1} \max_{x_{0}} \pb{x_1|x_{0}}= 0.5+0.5+0.6=1.6,\]
which indicates that $R_1=5/8$ is optimal for this example. However, we would like to mention that this example is special because it shows an instance where the bounds~\eqref{eq:theorem-outer} and \eqref{eq:theorem-inner} match. In general, for a choice of transition probabilities different from those given in \eqref{eq:example-tran-prob}, there might be a gap, as illustrated in Example~\ref{ex:symmetric_markov}.



\subsection{Proof of Theorem \ref{theorem:inner}}
\label{sec:proof_inner}
We will build on the previous example to describe the generalized scheme achieving the rate given
in \eqref{eq:theorem-inner}. The proposed coding scheme retrieves messages in the uncoded form, so we assume that $\cQ= \sP\left(\cN\right)$ in the remaining parts of this section. 

\noindent \textbf{Answer encoding function:}
The answer encoding function $\phi_t$ is given by 
\begin{equation}
\label{eq:answer-encoding}
  	A_t = \phi_t \left(Q_t,W_{\cN}\right) = W_{\cA} 
\end{equation}
for any $Q_t=\cA \in \sP\left(\cN\right)$. 

\noindent \textbf{Answer length function:}
The length of the answer is given by
\begin{equation*}
	\ell\left(Q_t\right) = |Q_t|\,L,
\end{equation*}
and the normalized average length is then given by
\begin{equation}
\label{eq:achievable-rate}
	\frac{1}{R_t} = \frac{\ell_t}{L} = \mathbb{E}\left[|Q_t|\right].
\end{equation}

\noindent \textbf{Query encoding function:}
At time $t$, suppose that the query $Q_t$ is a stochastic function of $X_t$, $X_{\tau}$ and $Q_{[t-1]}$. Recall that $\tau=\max \,\cB_t$, \Ie the last time privacy was ON. 
For any given $q_{[t-1]}$, we claim that there exists an encoding function $\wb{q_t|x_{\tau},x_t,q_{[t-1]}}$ giving
\begin{equation}
    \label{eq:expected_value}
	\mathbb{E}\left[|Q_t||q_{[t-1]}\right] \leq  \sum_{i=1}^{N} i \, \theta_i(q_{[t-1]}),
\end{equation} 
as well as satisfying two constraints, \Ie
\begin{equation}
\label{eq:achieve-two-decode}
	\pb{x_t,q_t|q_{[t-1]}}=0, \forall x_t \notin q_t,
\end{equation}
and
\begin{equation}
    \label{eq:privacy}
	I(Q_t;X_{\tau}|Q_{[t-1]}=q_{[t-1]}) = 0,
\end{equation}
Note that \eqref{eq:achieve-two-decode} guarantees the decodability from the answer encoding function $\phi_t$ given by \eqref{eq:answer-encoding}, and \eqref{eq:privacy} is a relaxed privacy constraint, where we recall the original privacy constraint $I \left(Q_t;X_{\cB_t}|Q_{[t-1]}\right) = 0$.

The following lemma justifies the existence of such a query encoding function.
\begin{lemma}
\label{lemma:algorithm}
For any given random variables $U, X \in \cN$, suppose that
\begin{equation}
    \begin{aligned}
        & \pb{X=x|U =u^{(x,1)}} \leq \pb{X=x|U =u^{(x,2)}} \leq \cdots \leq \pb{X=x|U =u^{(x,N)}}.
    \end{aligned}
\end{equation} 
Then, there exists a random variable $Y\in \sP(\cN)$ such that $Y$ is independent of $U$, $p(x,y)=0$ for $x \notin y$, and
\begin{equation}
\label{eq:lemma-upper}
    \mathbb{E}\left[|Y|\right] \leq  \sum_{i=1}^{N} i \, \theta_i,
\end{equation}
where $\theta_i=\min \left\{ 1,\sum\limits_{x \in \cN} \pb{X=x|U=u^{(x,i)}}\right\} - \min \left\{1,\sum\limits_{x \in \cN} \pb{X=x|U=u^{(x,i-1)}}\right\}$ for $i=1,\ldots,N$. 
\end{lemma}
\begin{proof}
We prove Lemma~\ref{lemma:algorithm} 
by designing a 
distribution $\pb{y,x|u}$ for any given distribution $\pb{x|u}$ satisfying the constraints  $Y\perp U$, $p(x,y)=0$ for $x \notin y$, and
\begin{equation*}
    \pb{|Y| \leq i} \geq \sum_{j=1}^{i} \theta_j, ~ i= 1,\ldots,N.
\end{equation*}
Moreover, we show that $\mathbb{E}\left[|Y|\right] \leq   \sum_{i=1}^{N} i \, \theta_i$ follows from the last constraint. The proof of Lemma~\ref{lemma:algorithm} is constructive,  \Ie we provide an algorithm that outputs the desired distribution. The details of the construction will be presented at the end of this section, and the justification of the algorithm and analysis of its complexity will be deferred to Appendix~\ref{app:proof-lemma-achievability}. 
\end{proof}

Before the detailed proof, we give the following reflections on the lemma.
\begin{enumerate}
\item This lemma generalizes the process we used to fill Table~\ref{table-example} for $N=3$ in Subsection~\ref{ex:achievability_N}. However, one may notice that the table therein contains about $N^2\,2^N$ entries, so any linear time approach such as filling them one by one will introduce an exponential blowup in complexity. Hence, the proof of the lemma  not only justifies the existence of an admissible $\pb{y,x|u}$, but also proposes a $\text{poly}(N)$ time algorithm to construct a $\pb{y,x|u}$ for any given distribution $\pb{x|u}$ to satisfy the constraints. 
\item If we treat each probability $\pb{y,x|u}$ for $x, u \in \cN$ and $y \in \sP\left(N\right)$ as a decision variable, we can see that both the objective function $\mathbb{E}[|Y|]$ and
two constraints, \Ie $Y$ is independent of $U$ and $p(x,y)=0$ for $x \notin y$, are linear, and hence the problem can  be indeed formulated as a linear programming problem with roughly
$N^2\,2^N$ variables and constraints, which makes the numerical solution impossible when $N$ goes large. The lemma here is aimed at finding a solution efficiently (avoid exponential overhead) and analytically (evaluate the objective value). 
More interpretations on this linear programming  perspective will be discussed in Section~\ref{sec:LP}. 

\end{enumerate}

For any given $q_{[t-1]}$, by letting $U \sim p_{X_{\tau}|q_{[t-1]}}$ and $X \sim p_{X_{t}|q_{[t-1]}}$ in Lemma~\ref{lemma:algorithm}, we can easily see that this lemma justifies the existence of  a query encoding function $\wb{q_t|x_{\tau},x_t,q_{[t-1]}}$ satisfying \eqref{eq:expected_value}, \eqref{eq:achieve-two-decode} and \eqref{eq:privacy}. 
The remaining piece to show is that the relaxed privacy constraint \eqref{eq:privacy} implies the desired privacy constraint \eqref{eq:privacy-1} for the given scheme,
\Ie $I(Q_t;X_{\tau}|Q_{[t-1]}) = 0$ implies $I \left(Q_t;X_{\cB_t}|Q_{[t-1]}\right) = 0$, which can be justified by Proposition~\ref{proposition:achieve-markov}. 
Therefore, we finish showing that for any given $q_{[t-1]}$, there exits an encoding function $\wb{q_t|x_{\tau},x_t,q_{[t-1]}}$ satisfying the decodability and the privacy constraint. Also, we know from Lemma~\ref{lemma:algorithm} that the encoding function $\wb{q_t|x_{\tau},x_t,q_{[t-1]}}$ yields 
\begin{equation*}
    \mathbb{E}\left[|Q_t||q_{[t-1]}\right] \leq  \sum_{i=1}^{N} i \, \theta_i(q_{[t-1]}).
\end{equation*} 
By averaging over all $q_{[t-1]}$, we have
\begin{equation}
\label{eq:achievable-rate-cardinality}
    \mathbb{E}\left[|Q_t|\right] \leq  \sum_{q_{[t-1]}} p(q_{[t-1]}) \sum_{i=1}^{N} i \, \theta_i(q_{[t-1]}),  
\end{equation}
which implies that $R_t$ is achievable (c.f.\eqref{eq:achievable-rate}) if 
\begin{equation}
    \frac{1}{R_t} \geq \sum_{q_{[t-1]}} p(q_{[t-1]}) \sum_{i=1}^{N} i \, \theta_i(q_{[t-1]}).
\end{equation}

\subsection{Constructive proof of Lemma~\ref{lemma:algorithm}}

First, let us recall some definitions and notation which will be used frequently in this section. For a fixed $x \in \cN$, suppose that
\begin{equation}
\label{eq:proof-order}
    \begin{aligned}
        & \pb{X=x|U =u^{(x,1)}} \leq \pb{X=x|U =u^{(x,2)}} \leq \cdots \leq \pb{X=x|U =u^{(x,N)}},
    \end{aligned}
\end{equation} 
where $u^{(x,i)}$ for $i=1,\ldots,N$ are $N$ distinct elements in $\cN$. Let
\begin{equation}
\label{eq:proof-sum}
  \lambda_i =  \sum_{x \in \cN} \pb{X=x|U=u^{(x,i)}},  
\end{equation}
and  
\begin{equation}
\label{eq:proof-diff}
    \theta_i = \min \{ 1,\lambda_i\} - \min \{ 1,\lambda_{i-1}\},
\end{equation}
where $\lambda_{0}$ is assumed to be $0$. Note that $\sum_{i=1}^{N} \theta_i = 1$.
Also, let
\begin{equation}
\label{eq:proof-maximum}
    \sigma = \max\{i: \lambda_i \leq 1 \}.  
\end{equation}


In this section, we will prove Lemma~\ref{lemma:algorithm} by designing a 
distribution $\pb{y,x|u}$ for any given distribution $\pb{x|u}$ satisfying the constraints  $Y\perp U$, $p(x,y)=0$ for $x \notin y$, and
\begin{equation}
\label{eq:lemma-probability}
    \pb{|Y| \leq i} \geq \sum_{j=1}^{i} \theta_j, ~ i= 1,\ldots,N.
\end{equation}
One can check that \eqref{eq:lemma-probability} yields
\begin{align*}
    \mathbb{E}\left[|Y|\right] & = \sum_{i=1}^N   i \, \pb{|Y| = i} 
     = \sum_{i=1}^N   \sum_{j=1}^{i}  \pb{|Y| = i}  
     = \sum_{j=1}^N   \sum_{i=j}^{N}  \pb{|Y| = i} 
     = \sum_{j=1}^N    \pb{|Y| \geq j} \\
    & = \sum_{j=1}^N  \left(1 - \pb{|Y| \leq j-1} \right)  
     \leq \sum_{j=1}^N  \left(1 - \sum_{i=1}^{j-1} \theta_{i}  \right)  
     = \sum_{j=1}^N   \sum_{i=j}^{N} \theta_{i}   
     = \sum_{i=1}^N   \sum_{j=1}^{i} \theta_{i}   \\
    & =  \sum_{i=1}^{N} i \, \theta_i,
\end{align*}
\Ie \eqref{eq:lemma-upper} to be proved in Lemma~\ref{lemma:algorithm}.

In particular, let $Z$ be a multiset $\left(\cN,m \right)$, where $\cN$ is the ground set and $m$ is the multiplicity function. The cardinality of the multiset $Z$ is the summation of multiplicities of all its element, \Ie
\begin{equation}
    |Z|= \sum_{x \in N} m(x).
\end{equation}
For example, given the ground set $\{a,b\}$ and the multiset $\{a,a,b\}$, the multiplicities of $a$ and $b$ are $m(a)=2$ and $m(b)=1$ respectively, and the cardinality of $|\{a,a,b\}|$ is $3$. For ease of notation, denote
\[\cZ = \left\{Z: Z \in \left(\cN,m \right), |Z| \leq N \right\},\]
\Ie the multiset whose elements are chosen from $\cN$ and whose cardinality is upper bounded by $N$.


We will prove that for any given $X$ and $U$, \Ie given any distribution $\pb{x|u}$ for $x, u \in \cN$, there exists a random variable $Z$ taking values in $\cZ$
such that $Z\perp U $, $p(x,z)=0$ for $x \notin z$, and 
\begin{equation}
\label{eq:lemma-upper-multiset}
    \pb{|Z| = i} = \theta_i,~\forall i=1,\ldots,\sigma+1, 
\end{equation}
Note that $\theta_i = 0$ for $i > \sigma+1$ from the definition  \eqref{eq:proof-diff}. By letting $Y=\text{Set}(Z)$, \Ie $Y$ is the corresponding set of the multiset $Z$,
we can easily see that if $Z\perp U $ and $p(x,z)=0$ for $x \notin z$, then $Y\perp U $ and $p(x,y)=0$ for $x \notin y$. Also, one can easily check that if \eqref{eq:lemma-upper-multiset} is satisfied, then \eqref{eq:lemma-probability} holds. Therefore, it is sufficient for us to justify the existence of such a $Z$ for any given $X$ and $U$. 

Now, we start the constructive proof,  \Ie for any given distribution $\pb{x|u}$, we will give an algorithm to construct some $Z$ satisfying that
\begin{equation}
    \pb{z,x}=0,~\forall x \notin z,
\end{equation}
and 
\begin{equation}
    \pb{z|u}=\pb{z|u'},~\forall z \in \cZ~\text{and}~u,u' \in \cN.
\end{equation}
Finally, we will show that the constructed $Z$ gives \eqref{eq:lemma-upper-multiset}, \Ie 
\begin{equation*}
    \pb{|Z| = i} = \theta_i,~\forall i=1,\ldots,\sigma+1.
\end{equation*}

\noindent \textbf{Input:} A distribution $\pb{x|u}$ for $x,u \in \cN$.

\noindent \textbf{Pre-calculation:}
\begin{enumerate}
\item \label{item:pre-first} 
For any given distribution $\pb{x|u}$, by sorting $\pb{x|u}$ for each $x \in \cN$, we can easily obtain parameters 
\[\left\{u^{(x,i)}, \lambda_i, \theta_i, \sigma: x \in \cN, i=1,\ldots,N\right\}\]
as defined in \eqref{eq:proof-order}-\eqref{eq:proof-maximum}. We will refer to these notations directly in the sequel. 

\item \label{item:pre-two} 
Then, we randomly pick a set of real numbers $\{\delta_j:j=1,\ldots,N\}$ such that
\begin{equation}
    \pb{X=j|U=u^{(j,\sigma)}} \leq \delta_j \leq \pb{X=j|U=u^{(j,\sigma+1)}}, ~\forall j\in \cN,
\end{equation}
and
\begin{equation}
    \sum_{j=1}^{N} \delta_j = 1.    
\end{equation}
The existence of such a set of $\{\delta_j : j=1,\ldots,N\}$ can be guaranteed by the definition of $\sigma$, since
\[\lambda_{\sigma} = \sum_{j=1}^{N}\pb{X=j|U=u^{(j,\sigma)}} \leq \sum_{j=1}^{N}\delta_j \leq \sum_{j=1}^{N}\pb{X=j|U=u^{(j,\sigma+1)}}= \lambda_{\sigma+1},\]
and $\lambda_{\sigma} \leq 1 < \lambda_{\sigma+1}$. 
\end{enumerate}
\emph{Specification:}
Here we specify a deterministic way of picking $\delta_j$ for $j=1,\ldots,N$. For notational simplicity, let 
$a_j = \pb{X=j|U=u^{(j,\sigma)}}$ 
and 
$b_j = \pb{X=j|U=u^{(j,\sigma+1)}}$
for $j=1,\ldots,N$. Then, provided two non-negative arrays $(a_1,\ldots,a_N)$ and $(b_1,\ldots,b_N)$ such that 
\[\sum_{j=1}^N a_j \leq 1 < \sum_{j=1}^N b_j,\]
our goal is to output an array $(\delta_1,\ldots,\delta_N)$ such that 
\[a_j \leq \delta_j \leq b_j, ~\forall j=1,\ldots,N,\]
and 
\[\sum_{j=1}^N \delta_j = 1.\]

We may choose $\delta_j$ sequentially and greedily. In particular, initialize $T = 0$. For $j=1,\ldots,N$, update $T$ by $T+\left(b_j - a_j \right)$. If  
\begin{equation*}
  T \leq 1- \sum_{j=1}^N a_j, 
\end{equation*}
then let $\delta_j = b_j$, otherwise let 
\begin{equation*}
    \delta_j = 1 - \sum_{k=1}^{j-1} b_j - \sum_{k=j+1}^{N} a_j
\end{equation*}
and $\delta_k = a_k$ for $k = j+1,\ldots,N$ to finish the process. 
Let $Q$ be an auxiliary $N \times N$ matrix which will be updated during the algorithm.
Also, let $Q^{-}_{i,j} = a$ denote $Q_{i,j} = Q_{i,j} - a$, \Ie subtracting $a$ from $Q_{i,j}$. 

\noindent \textbf{Initialization:}
Let
\begin{equation}
\label{eq:initialize}
    Q_{i,j}  =  \max\left\{\pb{X=j|U=i}-\delta_j, 0\right\}, \, i, j \in \cN.
\end{equation}

\noindent \textbf{Procedure:}
\begin{figure}[t]
    \centering
    \input{figure_proof}
    \caption{The rows represents $V_1,\dots,V_{\ell-1}$. A given row $i$ is divided, by \textit{boundaries}, into $e_i$ parts of different sizes, corresponding to $v_{i,1},\dots v_{i,e_i}$, \Eg $V_1$ is divided into $v_{1,1}$, $v_{1,2}$, and $v_{1,3}$. Moreover, rows are the same size in total to satisfy \eqref{eq:proof-sum-identical}. Then, every $\nu_k$ represents the number between two consecutive \textit{boundaries}.}
    \label{fig:proof}
\end{figure}
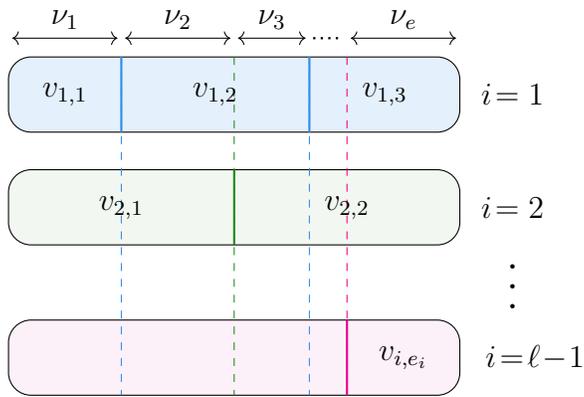
For $|Z|=\ell=1,\ldots,\sigma+1$, we consider the following process. For $x=1,\ldots,N$, identify $\{u^{(x,i)}: i =1, \ldots, \ell-1\}$. 
\begin{enumerate}
    \item For each $u^{(x,i)}$, we randomly choose a collection of pairs 
\begin{equation}
\label{eq:proof-choice}
 I_{i} \times V_{i} = \left\{\left(x_{i,j},v_{i,j}\right):j=1,2,\ldots \right\}   
\end{equation}
such that 
\begin{equation}
\label{eq:alg-t2}
  0 \leq v_{i,j} \leq Q_{u^{(x,i)},x_{i,j}},  
\end{equation}
and
\begin{equation}
\label{eq:alg-sum}
    \sum_{j} v_{i,j} = \min\left\{\delta_{x},  \pb{X=x|U=u^{(x,\ell)}}\right\}- \pb{X=x|U=u^{(x,\ell-1)}}.
\end{equation}
Note that the right-hand side of \eqref{eq:alg-sum}  only depends on $\ell$ and $x$ and is independent of $u^{(x,i)}$, which means that 
\begin{equation}
\label{eq:proof-sum-identical}
    \sum_{j=1} v_{1,j} = \cdots =  \sum_{j=1} v_{\ell-1,j},
\end{equation}
though the cardinality of $V_i$ for each $i$ may or may not be the same. For ease of notation, suppose that  
\[|I_{i}|=|V_{i}|= e_{i} \leq N.\]

After that, we update the matrix $Q$ by
\begin{equation}
\label{eq:alg-t3}
    Q_{u^{(x,i)},x_{i,j}}^{-} = v_{i,j}.
\end{equation}

It is clear from \eqref{eq:alg-t2} and \eqref{eq:alg-t3} that $Q$ is always non-negative, so the existence of such a collection $I_{i} \times V_{i}$ can be guaranteed if the following condition is satisfied
\begin{equation}
\label{eq:proof-existence}
        \sum_{k=1}^{N} Q_{u^{(x,i)},k} \geq \min\left\{\delta_{x},  \pb{X=x|U=u^{(x,\ell)}}\right\}- \pb{X=x|U=u^{(x,\ell-1)}},
\end{equation}
which will be verified in Appendix~\ref{app:proof-lemma-achievability}.

\emph{Specification:} We specify a deterministic way of choosing $I_i \times V_i$ under the assumption that \eqref{eq:proof-existence} holds. If the right-hand side of \eqref{eq:alg-sum} is zero, then one can simply choose $I_i \times V_i$ to be the empty set. If the right-hand side of \eqref{eq:alg-sum} is strictly positive, 
we initialize $T=0$ and $j=1$.
Then for $k=1,\ldots,N$ such that $ Q_{u^{(x,i)},k} >0$, 
if
\begin{equation}
  T + Q_{u^{(x,i)},k} <  \text{R.H.S of \eqref{eq:alg-sum}},  
\end{equation}
let $v_{i,j}=Q_{u^{(x,i)},k}$, $x_{i,j}=k$. Then increae $j$ by one and update $T$ by adding $Q_{u^{(x,i)},k}$ to it.  
Otherwise, let
\[ v_{i,j}= \text{R.H.S of \eqref{eq:alg-sum}} - T \] and $x_{i,j}=k$  to finish the process. 

\item
For fixed $\ell$ and $x$, given $I_{i}$ and $V_{i}$ for $i= 1,\ldots,\ell-1$, we randomly pick a collection of pairs $\left\{( \zeta_{k},\nu_{k}): k=1,2,\ldots \right\}$ such that 
\begin{equation}
\label{eq:proof-multiset}
    \zeta_k \in I_{1} \times I_{2} \times \cdots \times I_{\ell-1},
\end{equation}
and 
\begin{equation}
\label{eq:proof-multiset-value}
    \sum_{k: \zeta_k(i)=x_{i,j}} \nu_k = v_{i,j},~ \forall  1 \leq i \leq \ell-1 ~ \text{and}~ 1 \leq j \leq e_i,   
\end{equation}
where $\zeta_k(i)$ is the $i$-th element of $\zeta_k$. The existence of such a collection can be basically illustrated by Figure~\ref{fig:proof}. For notational simplicity, denote 
\[|\left\{( \zeta_{k},\nu_{k}): k=1,2,\ldots \right\}|=e_{x,\ell}.\]

\emph{Specification:}
We specify a deterministic way to construct such a collection $\left\{( \zeta_{k},\nu_{k}): k=1,2,\ldots, e \right\}$.
Let us initially push $\left(v_{1,1},v_{2,1},\ldots,v_{\ell-1,1} \right)$ and $\left(x_{1,1},x_{2,1},\ldots,x_{\ell-1,1} \right)$ into buffers $\mathrm{B}_v$ and $\mathrm{B}_x$, respectively. Let $\nu_1= \min \mathrm{B_v}$  and $\zeta_1=\mathrm{B}_x$.
Assume that the minimal value of $\mathrm{B}_v$ appears in the $m$-th position for some $m \in \{1,\ldots,\ell-1\}$, \Ie $v_{m,1}$ is the minimal. If the minimal is not unique, just randomly choose one.
We update $\mathrm{B}_v$ by subtracting $v_{m,1}$ from all elements in $\mathrm{B}_v$ and then push $v_{m,2}$ into the buffer to replace $v_{m,1}-v_{m,1}$, \Ie
\[\mathrm{B}_v=\left(v_{1,1}-v_{m,1},\ldots, v_{m,2},\ldots,v_{\ell-1,1}-v_{m,1} \right).\]
Also, update $\mathrm{B}_x$ by letting 
\[\mathrm{B}_x=\left(x_{1,1},\ldots, x_{m,2},\ldots,x_{\ell-1,1} \right).\]

Then, let $\nu_2= \min \mathrm{B}_v$ and $\zeta_2=\mathrm{B}_x$, and update $\mathrm{B}_v$ and $\mathrm{B}_x$ by the same process as stated above. Keep doing this repeatedly until all values $v_{i,j}$ for $1 \leq i \leq \ell-1$ and $1 \leq j \leq e_i$ have been dealt with. Note that \eqref{eq:proof-sum-identical} guarantees that the process ends properly. In this process, we deal with one $v_{i,j}$ every round, so we have
\begin{equation}
e_{x,\ell}=\sum_{i=1}^{\ell-1} e_i \leq (\ell-1)N.
\end{equation}

\item
For each $k=1,\ldots,e_{x,\ell}$, let $z_k=\{\zeta_k,x\}$. Then we let $\cA_{k,x,\ell}$ be a collection of tuples defined as follows:
\begin{equation}
\label{eq:proof-arguments}
\begin{aligned}
    \cA_{k,x,\ell} = & \left\{\left(\bar{z},\bar{x},\bar{u}\right): \bar{z}=z_k, \bar{x}=\zeta_k(i), \bar{u}=u^{(x,i)}, i =1,\ldots,\ell-1 \right\} \\
    & ~~~~~ \bigcup \left\{\left(\bar{z},\bar{x},\bar{u}\right):\bar{z}=z_k,\bar{x}=x,\bar{u} \in \cN\setminus \left\{u^{(x,i)}:i=1,\ldots,\ell-1 \right\} \right\},
\end{aligned}
\end{equation}
where $|\cA_{k,x,\ell}|=N$. 
To avoid ambiguity in the following discussion, denote 
\begin{equation}
\label{eq:proof-nu}
  \nu_{k,x,\ell} = \nu_k.  
\end{equation}


\item
For a fixed $\ell$, denote
\begin{equation}
\label{eq:proof-merge-set}
  \cA_{\ell} = \bigcup_{1 \leq x \leq N} \bigcup_{1 \leq k \leq e_{x,\ell}}   \cA_{k,x,\ell},  
\end{equation}
and for any $\left(\bar{z},\bar{x},\bar{u}\right) \in \cA_{\ell}$, let
\begin{equation}
\label{eq:assign}
    \qb{\bar{z},\bar{x},\bar{u}} = 
        \sum_{x=1}^{N} 
         \sum_{k:\left(\bar{z},\bar{x},\bar{u}\right) \in \cA_{k,x,\ell}}  \nu_{k,x,\ell},
\end{equation}
where $k = 1,\ldots, e_{x,\ell}$.

\end{enumerate}

\noindent \textbf{Output:} 
The output of the algorithm is $\mathrm{OUT}=\left\{\cA_{\ell},\qb{\cA_{\ell}}: \ell=1,\ldots,\sigma+1\right\}$. Later, we will see that this pair indeed stores the non-zero valued arguments and corresponding values of $\pb{z,x|u}$, \Ie
\begin{equation}
    \pb{z,x|u} =
    \begin{cases}
            \qb{z,x,u}, & (z,x,u) \in  \cA,  \\
            0, & \text{otherwise},
    \end{cases}
\end{equation}
where $\cA:= \cup_{\ell} \cA_{\ell}$. Note that $\cA_{\ell}$ are disjoint with each other since $|z|=\ell$ for any $(z,x,u) \in \cA_{\ell}$ from~\eqref{eq:proof-multiset}.

For the better illustration, we summarize the constructive proof in Algorithm~\ref{alg:algorithm}.
\begin{algorithm}
\caption{}
\label{alg:algorithm}
\begin{algorithmic}[1]
\REQUIRE A given distribution $\pb{x|u}$ for $x,u \in \cN$
\ENSURE The non-zero valued arguments $\cA=\left\{(z,x,u): x, u \in \cN, z \in \cZ, \pb{z,x|u} > 0 \right\}$ and probabilities $\qb{\cA} =\left\{\pb{z,x|u}: (z,x,u) \in \cA \right\}$ for a distribution $\pb{z,x|u}$ such that $\pb{z|u}=\pb{z}$ and $\pb{z,x|u} = 0$ for any $x \notin z$
\STATE Pre-calculation
\STATE Initialize 
\FOR{$\ell=1,\ldots,\sigma+1$} 
\FOR{$x \in \cN$} 
\FOR{$u \in \{u^{(x,i)}:i=1,\ldots,\ell-1\}$}

\STATE Find a collection of pairs $I_i \times V_i$ satisfying \eqref{eq:alg-t2} and \eqref{eq:alg-sum}

\ENDFOR
\STATE Given $\left\{I_i \times V_i:i=1,\ldots, \ell-1\right\}$, find a collection of pairs $\left\{(\zeta_{k},\nu_{k}): k=1,2,\ldots , e_{x,\ell}\right\}$ satisfying  
\eqref{eq:proof-multiset} and \eqref{eq:proof-multiset-value}
 
\STATE Obtain $\left\{\cA_{k,x,\ell}, \nu_{k,x,\ell}: k=1,\ldots, e_{x,\ell} \right\}$ from \eqref{eq:proof-arguments} and \eqref{eq:proof-nu}

\ENDFOR

\STATE Merge $\left\{\cA_{k,x,\ell}: x \in \cN, k=1,\ldots, e_{x,\ell} \right\}$ to obtain $\cA_{\ell}$ and corresponding values $\qb{\cA_{\ell}}$ from \eqref{eq:proof-merge-set}
and \eqref{eq:assign}
\ENDFOR

\STATE  $\mathrm{OUT}=\left\{\cA_{\ell},\qb{\cA_{\ell}}: \ell=1,\ldots,\sigma+1\right\}$
\end{algorithmic}
\end{algorithm}

\section{Linear Programming Perspective}
\label{sec:LP} Inspired by the proposed scheme in the last section, we restrict our discussion to uncoded queries. Then the key step is to design a query encoding function $\wb{q_t|x_{\tau},x_t,q_{[t-1]}}$, that minimizes the download cost $\mathbb{E}\left[|Q_t|\right]$ subject to two constraints, 
 \Ie the decodability constraint \eqref{eq:achieve-two-decode} and a relaxed privacy constraint \eqref{eq:privacy} (protecting the last time when privacy was ON).
 
 For any given $q_{[t-1]}$, or more precisely given the input distribution $\pb{x_t|x_{\tau},q_{[t-1]}}$, 
the problem can then be alternatively formulated as a linear programming (LP) instance as follows,
\begin{equation}
\label{eq:LP}
    \begin{aligned}
    & \underset{\pb{q_t|x_{\tau},x_t,q_{[t-1]}}}{\text{minimize}}
    & & \mathbb{E}\left[|Q_t||q_{[t-1]}\right] = \sum_{q_t} \pb{q_t|q_{[t-1]}}|q_t| &\\
    & \text{subject to}
    & & \pb{x_t,q_t|q_{[t-1]}} =0, \ x_t \notin q_t,  & \text{(decodability)}\\
    & & & \pb{q_t|x_{\tau},q_{[t-1]}} = \pb{q_t|q_{[t-1]}}. & \text{(relaxed privacy)}\\
    \end{aligned}
\end{equation}

This linear programming problem has
$N^2\,2^N$ variables and $(N+2) N\,2^{N-1}$ constraints, \Ie each probability $\pb{q_t,x_t,x_{\tau}|q_{[t-1]}}$ is a variable where $x_t,x_{\tau} \in \cN$ and $q_t \in \sP{(\cN)}$.
The scale of the problem is intractable in complexity with any generic linear programming solver, for instance Vaidya's algorithm \cite{Vaidya} gives  $\cO\left(\left(N^2\,2^{N}\right)^{2.5}\right)$. 

One possible strategy dealing with the complexity issue is to impose a restriction on the cardinality of $q_t$, \Ie $|q_t|$ is chosen from $\{1,2,\ldots,c,N\}$ where $c$ is a constant and $N$ is included to guarantee the problem is feasible. 
\begin{equation}
\label{eq:LP2}
    \begin{aligned}
    & \underset{\pb{q_t|x_{\tau},x_t,q_{[t-1]}}}{\text{minimize}}
    & & \mathbb{E}\left[|Q_t||q_{[t-1]}\right] = \sum_{q_t} \pb{q_t|q_{[t-1]}}|q_t| &\\
    & \text{subject to}
    & & \pb{x_t,q_t|q_{[t-1]}} =0, \ x_t \notin q_t,  & \\
    & & & \pb{q_t|x_{\tau},q_{[t-1]}} = \pb{q_t|q_{[t-1]}}, & \\
    & & & |q_t| \in \{1,2,\ldots,c,N\}. & 
    \end{aligned}
\end{equation}
In this way, the number of variables drops dramatically as  the alphabet of $q_t$ is reduced from $2^N$ to the order of $N^c$, \Ie setting $\pb{q_t|x_{\tau},x_t,q_{[t-1]}}=0$ for $|q_t|=\{c+1,\ldots,N-1\}$.  Then, the LP instance roughly has $N^{c+2}$ variables, which makes solving the problem numerically possible. 
For instance if we choose  $c=1$, \Ie the user either downloads the message he wants or all messages on the server, we can obtain the optimal value to \eqref{eq:LP2}, which is
\begin{equation}
\label{eq:LP-extreme}
     \mathbb{E}\left[|Q_t||q_{[t-1]}\right] =  \theta_1(q_{[t-1]}) 
     + N \left(1- \theta_1(q_{[t-1]})\right),  
\end{equation} 
where
$\theta_1$ was previously defined (c.f.\eqref{eq:def_theta}) to be 
\[\sum_{x \in \cN} \min_{x_{\tau}} \pb{X_t=x|X_{\tau} = x_{\tau},Q_{[t-1]} = q_{[t-1]}}.\]

Instead of attempting to solve the linear programming problem numerically, Lemma~\ref{lemma:algorithm} in the last section actually identifies a feasible solution to the problem \eqref{eq:LP} \emph{efficiently}, and bounds the objective $\mathbb{E}\left[|Q_t||q_{[t-1]}\right]$ \emph{analytically}, \Ie a feasible solution attains an objective such that 
\begin{equation}
\label{eq:LP-feasible}
    \mathbb{E}\left[|Q_t||q_{[t-1]}\right] \leq  \sum_{i=1}^{N} i \, \theta_i(q_{[t-1]}).
\end{equation}
One can easily see that \eqref{eq:LP-feasible} outperforms \eqref{eq:LP-extreme}.  

A helpful observation here is that any algorithmic tractable solution should only visit a small proportion of the power set, i.e, the support set of the query $q_t$. Otherwise, since the power set is exponentially large, it will introduce an exponential overhead for configuring the probabilities $p\left(q_t|x_{\tau},x_t,q_{[t-1]}\right)$ for $x_{\tau}, x_t \in \mathcal{N}$ and $q_t \in \mathscr{P}(\mathcal{N})$. 

\section{Proof of Tightness for $N=2$ in Theorem \ref{theorem:two}}
\label{sec:proof_for_2}
In this section, we revisit the case  $N=2$, which was first studied in \cite{Naim_2019}. As previously stated, we will show the bounds obtained in Theorem~\ref{theorem:outer} and Theorem~\ref{theorem:inner} are tight for the case $N=2$. 
We will give an alternate proof to the specially designed one for $N=2$ presented in \cite{Naim_2019}, which   relies on the general results presented in Theorem~\ref{theorem:outer} and Theorem~\ref{theorem:inner}.

Before starting the proof we discuss some consequences of Theorem \ref{theorem:two}. We have the following observations.
\begin{itemize}
	\item If $F_t=\text{ON}$, then $\tau = t$ from the definition of $\tau$, then $\frac{1}{R_t} \geq 2$. This means that it is necessary to download both messages, which is consistent with the well-known result for the single server PIR \cite{Chor_1995}.
	\item If $F_t=\text{OFF}$, it is possible for the user to download less than two messages since $0 \leq \alpha+\beta \leq 2$. We can see that the rate as a function of $\alpha$ and $\beta$ is symmetric around $\alpha+\beta=1$. When $\alpha+\beta=1$, the Markov chain is independent, \Ie the user's requests are independent, the user can directly ask for the desired message, and the rate is $R_t=1$ (maximum).  When $\alpha=\beta=0$ or $\alpha=\beta=1$, \Ie the Markov chain is not ergodic, the user is required to ask for both messages, and then the rate is $R_t=1/2$ (minimum). Another observation is that when the Markov chain is ergodic, the rate goes to $1$ when $t-\tau$ goes to infinity. Intuitively, as $t-\tau$ grows, the information carried by $X_t$ about $X_{\tau}$ decreases, so the user can eventually directly ask for the desired message without being concerned about leaking information about $X_t$.  
\end{itemize}

\subsection{Converse}
\label{sec:two-converse}	
It is sufficient to show that the right-hand side of \eqref{eq:corollary} equals to $1+ |1- \alpha -\beta|^{t-\tau}$. We first write the right-hand side of \eqref{eq:corollary} explicitly in terms of $\alpha$ and $\beta$. If $\alpha + \beta = 0$, then $\alpha=\beta=0$, and we have 
\begin{equation*}
 P^{t-\tau} = P = 
\begin{bmatrix}
     1  &  0         \\
       0  & 1 
\end{bmatrix},
\end{equation*} 
which yields
\begin{equation}
 \label{eq:two-identity}
 \sum_{x_t \in \cN} \max_{x_{\tau} \in \cN} p\left(x_t|x_{\tau}\right) = 2.
\end{equation} 

If $\alpha + \beta \neq 0$, $p\left(x_t|x_{\tau}\right)$ is given by the transition matrix $ P^{t-\tau}$, \Ie
\begin{equation}
 P^{t-\tau} = \frac{1}{\alpha+ \beta}
\begin{bmatrix}
     \beta + \alpha(1- \alpha - \beta)^{t-\tau}  
       &  \alpha- \alpha(1- \alpha - \beta)^{t-\tau}         \\
       \beta - \beta(1- \alpha - \beta)^{t-\tau}  & \alpha + \beta (1- \alpha - \beta)^{t-\tau} 
\end{bmatrix}.
\end{equation} 
Then, we have
\[\sum_{x_t} \max_{x_{\tau}} p\left(x_t|x_{\tau}\right) = 
\begin{cases}
1 + (1- \alpha - \beta)^{t-\tau},  &  (1- \alpha - \beta)^{t-\tau}   \geq   0,\\
1 - (1- \alpha - \beta)^{t-\tau},  &  (1- \alpha - \beta)^{t-\tau}   <  0,
\end{cases}
\]
which can also be written as 
\begin{equation}
\label{eq:two-nonidentity}
	\sum_{x_t} \max_{x_{\tau}} p\left(x_t|x_{\tau}\right) = 1+ |1- \alpha - \beta|^{t-\tau}.
\end{equation}

By combining \eqref{eq:two-identity} and \eqref{eq:two-nonidentity}, we get that $\sum_{x_t} \max_{x_{\tau}} p\left(x_t|x_{\tau}\right) = 1+ |1- \alpha - \beta|^{t-\tau}$
for any given $\alpha$ and $\beta$. Therefore, we have
\[ \frac{1}{R_t} \geq \sum_{x_t} \max_{x_{\tau}} p\left(x_t|x_{\tau}\right) = 1+ |1- \alpha - \beta|^{t-\tau}, \]
which completes the converse proof.

\subsection{Achievability}
\label{sec:two-achievability}
From Theorem~\ref{theorem:inner}, we know that the rate $R_t$ is achievable if 
\begin{equation}
\label{eq:achieve-two-rate-1}
	\frac{1}{R_t} \geq \sum_{q_{[t-1]}} p(q_{[t-1]}) \sum_{i=1}^{N} i \, \theta_i(q_{[t-1]}).
\end{equation}
Since $\lambda_1(q_{[t-1]}) \leq 1$ and $\lambda_2(q_{[t-1]}) \geq 1$ for $N=2$, \eqref{eq:achieve-two-rate-1} can be rewritten as  
\begin{equation}
\label{eq:achieve-two-rate}
	\frac{1}{R_t} \geq \sum_{q_{[t-1]}} p(q_{[t-1]}) \left( 2 - \sum_{x_t} \min_{x_{\tau} } \pb{x_t|x_{\tau},q_{[t-1]}}   \right).
\end{equation}

In this subsection, we will express the right-hand side of \eqref{eq:achieve-two-rate} explicitly in terms of $\alpha$ and $\beta$, and we will show that it is exactly equal to $1+ |1- \alpha -\beta|^{t-\tau}$, as given in \eqref{eq:theorem-two}. Also, we will explicitly illustrate the encoding function $\wb{q_t|x_t,x_{\tau},q_{[t-1]}}$, which is exactly the same as the one presented in~\cite{Naim_2019}.

From the discussion in Section~\ref{section:inner}, we 
can 
infer that the query encoding function $\wb{q_t|x_t,x_{\tau},q_{[t-1]}}$ is given by
\begin{equation}
\label{eq:encoding-function-two}
\wb{q_t|x_t,x_{\tau},q_{[t-1]}} = 
	\begin{cases}
		\frac{\pi\left(x_t,q_{[t-1]} \right)}{\pb{x_t|x_{\tau},q_{[t-1]}}}, & |q_t|= 1, \\
		1- \frac{\pi\left(x_t,q_{[t-1]} \right)}{\pb{x_t|x_{\tau},q_{[t-1]}}}, & |q_t| = 2,
	\end{cases}
\end{equation}
where $\pi\left(x_t,q_{[t-1]} \right)$ is defined by
\[\pi\left(x_t,q_{[t-1]} \right): = \min_{x_{\tau} \in \{1,2\}} \pb{x_t|x_{\tau},q_{[t-1]}}.\]

Since $q_t \neq \bar{x_t}$ is always true (c.f.\eqref{eq:achieve-two-decode}), where $\bar{x}_t$ is the complement of $x_t$ in the set $\{1,2\}$,  \eqref{eq:encoding-function-two} is well-defined for any $q_t \in \left\{\{1\},\{2\},\{1,2\}\right\}$. As consequences, 
\begin{enumerate}
\item 
	When $F_t=\text{ON}$, $\tau= t$ by definition, and 
	\begin{equation}
	\label{eq:achieve-two-on}
		\min_{x_{\tau}}\pb{x_t|x_{\tau},q_{[t-1]}} = \min_{x'_{t}}\pb{x_t|x'_{t},q_{[t-1]}} =  0.
	\end{equation}
	This immediately implies that
	\begin{equation}
	\label{eq:ON}
	\wb{q_t|x_t,x_{\tau},q_{[t-1]}} = 
		\begin{cases}
			0, & |q_t|= 1, \\
			1, & |q_t| = 2,
		\end{cases}
	\end{equation} 
	for any $x_t$ and $q_{[t-1]}$, which means that the user will always download two messages when $F_t=\text{ON}$, \Ie
	\begin{equation}
	\label{eq:achieve-two-on-query}
		\pb{|Q_t|=2} = 1.
	\end{equation}

\item When $F_t=\text{OFF}$, $\tau \neq t$ by definition. Let
\begin{equation}
	\hat{x}_{\tau}(x_t,q_{[t-1]}) = \arg\min_{x_{\tau}} \pb{x_t|x_{\tau},q_{[t-1]}}
\end{equation}
for any $x_t$ and $q_{[t-1]}$. For notational simplicity, $\hat{x}_{\tau}(x_t,q_{[t-1]})$ will be written as $\hat{x}_{\tau}$ when $x_t$ and $q_{[t-1]}$ are clear from context. 
As such, we can see that
\begin{equation}
\label{eq:achieve-two-encf-min}
\wb{q_t|x_t,\hat{x}_{\tau},q_{[t-1]}} = 
\begin{cases}
		1, & |q_t|= 1, \\
		0, & |q_t| = 2.
	\end{cases}
\end{equation}
	\begin{itemize}
		\item If $\hat{x}_{\tau}(x_t,q_{[t-1]})$ is unique, since $X_{\tau}$ and $X_{t}$ take values in the binary alphabet, it is easy to check that 
	\begin{equation}
	 	\hat{x}_{\tau}(x_t,q_{[t-1]}) \neq \hat{x}_{\tau}(\bar{x_t},q_{[t-1]})
	\end{equation} 
	for any given $q_{[t-1]}$. This implies that $x_{\tau}$ and $x_{t}$ can be determined from each other provided that $|q_t|=2$. 
	In particular, assume that $|q_{t-1}|=2$, which implies that $F_{t-1}=\text{OFF}$. We know that $x_{\tau}$ and $x_{t-1}$ can be determined by each other provided that $|q_{t-1}|=2$, and hence we can easily obtain that
	\begin{equation}
	\label{eq:achieve-two-p1}
		\sum_{x_t} \pi\left(x_t,q_{[t-1]} \right) = \sum_{x_t} \min_{x_{\tau}}\pb{x_t|x_{\tau},q_{[t-1]}} = \sum_{x_t} \min_{x_{t-1}}\pb{x_t|x_{t-1}}.
	\end{equation}
	Correspondingly, we have
	\begin{align}
		\pb{|q_{t}|=1||q_{t-1}|=2}
		& =  \sum_{q_{[t-2]}} \sum_{x_{\tau}}\pb{x_{\tau},q_{[t-2]}|q_{t-1}} \sum_{x_t}\pb{q_t,x_t|x_{\tau},q_{[t-1]}} \nonumber \\
		& = \sum_{q_{[t-2]}} \sum_{x_{\tau}}\pb{x_{\tau},q_{[t-2]}|q_{t-1}} \sum_{x_t} \pi\left(x_t,q_{[t-1]} \right) \nonumber \\
		& = \sum_{q_{[t-2]}} \sum_{x_{\tau}}\pb{x_{\tau},q_{[t-2]}|q_{t-1}} 
		\sum_{x_t} \min_{x_{t-1}}\pb{x_t|x_{t-1}} \nonumber \\
		& = \sum_{x_t} \min_{x_{t-1}}\pb{x_t|x_{t-1}}. \label{eq:two-q-two}
	\end{align}
		\item If $\hat{x}_{\tau}(x_t,q_{[t-1]})$ is not unique, \Ie $\pb{x_t|x_{\tau},q_{[t-1]}} = \pb{x_t|\bar{x}_{\tau},q_{[t-1]}}$, then we can easily see that 
		\begin{equation}
		\wb{q_t|x_t,x_{\tau},q_{[t-1]}} = 
		\begin{cases}
				1, & |q_t|= 1, \\
				0, & |q_t| = 2,
			\end{cases}
		\end{equation}
		for any $x_{\tau} \in \{1,2\}$. In particular, if $|q_{t-1}| = 1$, implying that $F_{t-1}=\text{OFF}$, then $\tau < t-1$ by definition, and hence from the fact $x_{t-1}=q_{t-1}$ when $|q_{t-1}| = 1$,
		we can obtain that 
		\[\hat{x}_{\tau}(x_t,q_{[t-1]}) = \arg\min_{x_{\tau}} \pb{x_t|x_{\tau},q_{[t-1]}} = \arg\min_{x_{\tau}} \pb{x_t|x_{t-1}},~ x_{t-1} = q_{t-1}.\]
		We can easily see that $\hat{x}_{\tau}(x_t,q_{[t-1]})$ is not unique in this case, which implies that
		\begin{equation}
		\label{eq:achieve-two-p2}
		   	\sum_{x_t} \pi\left(x_t,q_{[t-1]} \right) = \sum_{x_t} \min_{x_{\tau}}\pb{x_t|x_{\tau},q_{[t-1]}} 
			=\sum_{x_t}  \pb{x_t|x_{t-1}} = 1,
		\end{equation}   
		and
		\begin{equation}
		\label{eq:two-q-one}
			\pb{|q_{t}|=1||q_{t-1}|=1} = 1.
		\end{equation}		
	\end{itemize}
\end{enumerate}

In summary, 
\begin{enumerate}
\item When $F_t=\text{ON}$, we have 
	\[\pi\left(x_t,q_{[t-1]}\right)  =  0,\] 
	and hence by substituting in \eqref{eq:achieve-two-rate}, we can see that
	\begin{equation}
		\label{eq:achieve-two-temple-tt1}
		\frac{1}{R_t} \geq \sum_{q_{[t-1]}} p(q_{[t-1]}) \left( 2 - \sum_{x_t} \min_{x_{\tau} } \pb{x_t|x_{\tau},q_{[t-1]}}   \right) = 2 \sum_{q_{[t-1]}} p(q_{[t-1]}) = 2.
	\end{equation}
\item When $F_t=\text{OFF}$, we have from \eqref{eq:achieve-two-p1} and \eqref{eq:achieve-two-p2} that 
\begin{equation}
	\sum_{x_t} \pi\left(x_t,q_{[t-1]} \right) = 
	\begin{cases}
		1, & |q_{t-1}|=1, \\
	\sum_{x_t} \min_{x_{t-1}}\pb{x_t|x_{t-1}}, & |q_{t-1}|=2. 
	\end{cases}
\end{equation}	

By substituting in \eqref{eq:achieve-two-rate}, we get that
\begin{align}
	\frac{1}{R_t} 
	& \geq \sum_{q_{[t-1]}} p(q_{[t-1]}) \left( 2 - \sum_{x_t} \min_{x_{\tau} } \pb{x_t|x_{\tau},q_{[t-1]}}   \right) \nonumber \\
	& = \pb{|Q_{t-1}|=1} \times 2 +  \pb{|Q_{t-1}|=2} \left( 2-   \sum_{x_t} \min_{x_{t-1}}\pb{x_t|x_{t-1}}\right)  \nonumber \\
	& = 2 - \pb{|Q_{t-1}|=2} \left( \sum_{x_t} \min_{x_{t-1}}\pb{x_t|x_{t-1}}\right). \label{eq:achieve-two-temple-t1}
\end{align}

From \eqref{eq:two-q-two} and \eqref{eq:two-q-one}, we can easily see the following proposition.
\begin{proposition}
\label{lemma:two-query-markov}
	$\left\{|Q_i|: \tau \leq i \leq t\right\}$ forms a Markov chain, and the transition matrix is given by
\begin{equation}
\label{eq:two-query-markov}
	\begin{bmatrix}
		1 & 0 \\
		\sum_{x_t}  \min_{x_{t-1}}  \pb{x_t|x_{t-1}} & 1- \sum_{x_t}  \min_{x_{t-1}}  \pb{x_t|x_{t-1}}
	\end{bmatrix}.
\end{equation}
\end{proposition}
	
From the definition of $\tau$, we know that $F_{\tau}=\text{ON}$, and hence $p\left(|Q_{\tau}|=2\right) = 1$ from \eqref{eq:achieve-two-on-query}. Then from Proposition~\ref{lemma:two-query-markov}, we have
\begin{align*}
	p\left(|Q_{t-1}|=2\right) = 
	 \left(1-\sum_{x_t} \min_{x_{t-1}} \pb{x_t|x_{t-1}}\right)^{t-1-\tau}. 
\end{align*}

Hence, \eqref{eq:achieve-two-temple-t1} can be written as 
\begin{align}
	\frac{1}{R_t} 
	& \geq 2 - \pb{|Q_{t-1}|=2} \left( \sum_{x_t} \min_{x_{t-1}}\pb{x_t|x_{t-1}}\right) \nonumber\\
	& = 2 -  \left( \sum_{x_t} \min_{x_{t-1}}\pb{x_t|x_{t-1}}\right)\left(1-\sum_{x_t} \min_{x_{t-1}} \pb{x_t|x_{t-1}}\right)^{t-1-\tau}. \label{eq:achieve-two-temple-tt2}
\end{align}

\end{enumerate}

By substituting $\alpha$ and $\beta$ in \eqref{eq:achieve-two-temple-tt1} and \eqref{eq:achieve-two-temple-tt2}, we can easily check that both inequalities can be written as  
\[\frac{1}{R_t} \geq    |1- \alpha- \beta|^{t-\tau}. \]
Moreover, one can also check that the encoding function $\wb{q_t|x_t,x_{\tau},q_{[t-1]}}$ given in \eqref{eq:encoding-function-two} can be expressed as in Table~\ref{table:Qt} which was first presented in \cite{Naim_2019}.

\appendices

\section{Proof of Proposition~\ref{proposition:achieve-markov}}
\label{appendix:proposition-achieve-markov}
It is clear that we need to show that 
\[I \left(X_{\tau};Q_t|Q_{[t-1]}\right) = 0\]
implies that 
\[	I \left(X_{\cB_t};Q_t|Q_{[t-1]}\right) = 0.\]

Consider
\begin{align*}
	I \left(X_{\cB_t};Q_t|Q_{[t-1]}\right)
	& = I \left(X_{\tau};Q_t|Q_{[t-1]}\right) + I \left(X_{\cB_t\backslash\{\tau\}};Q_t|X_{\tau},Q_{[t-1]}\right) \\
	& \utag{a}{\leq} I \left(X_{\tau};Q_t|Q_{[t-1]}\right) + I \left(X_{\cB_t\backslash\{\tau\}};X_t,S_t|X_{\tau},Q_{[t-1]}\right) \\
	& \utag{b}{=} I \left(X_{\tau};Q_t|Q_{[t-1]}\right) + I \left(X_{\cB_t\backslash\{\tau\}};X_t|X_{\tau},Q_{[t-1]}\right), 
\end{align*}
where \uref{a} follows because $Q_t$ is a function of $\{X_{\tau},X_t,S_t,Q_{[t-1]}\}$, and \uref{b} follows because the local randomness is generated according to $p_{X_t,X_{\tau},Q_{[t-1]}}$.

It remains to show that 
\[I \left(X_{\cB_t\backslash\{\tau\}};X_t|X_{\tau},Q_{[t-1]}\right) = 0,\]
which can be justified as follows:
\begin{align*}
	I \left(X_{\cB_t\backslash\{\tau\}};X_t|X_{\tau},Q_{[t-1]}\right)
	& = H \left(X_{\cB_t\backslash\{\tau\}}|X_{\tau},Q_{[t-1]}\right) - H \left(X_{\cB_t\backslash\{\tau\}}|X_t,X_{\tau},Q_{[t-1]}\right)  \\
	& \leq H \left(X_{\cB_t\backslash\{\tau\}}|X_{\tau},Q_{[t-1]}\right) - H \left(X_{\cB_t\backslash\{\tau\}}|X_t,X_{\tau},Q_{[\tau-1]},X_{[\tau:t-1]},S_{[\tau:t-1]}\right) \\
	& \utag{c}{=} H \left(X_{\cB_t\backslash\{\tau\}}|X_{\tau},Q_{[t-1]}\right) - H \left(X_{\cB_t\backslash\{\tau\}}|X_t,X_{\tau},Q_{[\tau-1]},X_{[\tau:t-1]}\right) \\
	& = H \left(X_{\cB_t\backslash\{\tau\}}|X_{\tau},Q_{[t-1]}\right) - H \left(X_{\cB_t\backslash\{\tau\}}|X_{\tau},Q_{[\tau-1]},X_{[\tau+1:t]}\right) \\
	& \utag{d}{=} H \left(X_{\cB_t\backslash\{\tau\}}|X_{\tau},Q_{[\tau-1]}\right) - H \left(X_{\cB_t\backslash\{\tau\}}|X_{\tau},Q_{[\tau-1]}\right) \\
	& = 0,
\end{align*}
where \uref{c} follows because $S_{[\tau:t-1]}$ is independent of $X_{\cB_t\backslash\{\tau\}}$ given $\left\{X_{[\tau:t-1]},Q_{[t-1]} \right\}$, and \uref{d} follows from the markovity of $\{X_t:t \in \mathbb{N}\}$.

\section{Proof of Corollary~\ref{corollary:outer}}
\label{appendix:corollary-outer}
Recall that the inequality that needs to be shown is 
\begin{equation}
\label{eq:appendix-corollary}
    \sum_{q_{[t-1]}} \pb{q_{[t-1]}} \sum_{x_t} \max_{x_{\tau}} \pb{x_t|x_{\tau},q_{[t-1]}} \geq \sum_{x_t} \max_{x_{\tau}} \pb{x_t|x_{\tau}}.
\end{equation}

Since
\begin{align*}
    \sum_{q_{[\tau]}} \pb{q_{[\tau]}} \sum_{x_t} \max_{x_{\tau}} \pb{x_t|x_{\tau},q_{[\tau]}} = \sum_{x_t} \max_{x_{\tau}} \pb{x_t|x_{\tau}},
\end{align*}
where the equality follows because $Q_{[\tau]}$ is a stochastic function of $X_{[\tau]}$, and hence $Q_{[\tau]}$ is independent of $X_t$ given $X_{\tau}$, \Ie
\begin{equation}
\label{eq:appendix-A-Markov}
    Q_{[\tau]} \rightarrow X_{\tau} \rightarrow X_t,
\end{equation}
due to the Markovity of $\{X_i:i \in \mathbb{N}\}$.
Thus, we can easily see that \eqref{eq:appendix-corollary} holds for $t=\tau+1$.

For any $i \in [\tau+2:t]$, consider
\begin{align}
    \sum_{q_{[i-1]}} \pb{q_{[i-1]}} \sum_{x_t} \max_{x_{\tau}} \pb{x_t|x_{\tau},q_{[i-1]}}
    & \geq \sum_{q_{[i-2]}} \pb{q_{[i-2]}} \sum_{x_t} \max_{x_{\tau}} \sum_{q_{i-1}} \pb{q_{i-1}|q_{[i-2]}} \pb{x_t|x_{\tau},q_{[i-1]}} \nonumber \\
    & \utag{a}{=} \sum_{q_{[i-2]}} \pb{q_{[i-2]}} \sum_{x_t} \max_{x_{\tau}} \sum_{q_{i-1}} \pb{q_{i-1}|q_{[i-2]},x_{\tau}} \pb{x_t|x_{\tau},q_{[i-1]}} \nonumber \\
    & = \sum_{q_{[i-2]}} \pb{q_{[i-2]}} \sum_{x_t} \max_{x_{\tau}}  \pb{x_t|x_{\tau},q_{[i-2]}}, \label{eq:appendix-corollary-t1}
\end{align}
where \uref{a} follows from the privacy at time $i-1$.

Since \eqref{eq:appendix-corollary-t1} holds for any $i \in [\tau+2:t]$, we can easily obtain that
\[\sum_{q_{[t-1]}} \pb{q_{[t-1]}} \sum_{x_t} \max_{x_{\tau}} \pb{x_t|x_{\tau},q_{[t-1]}} \geq  \sum_{q_{[\tau]}} \pb{q_{[\tau]}} \sum_{x_t} \max_{x_{\tau}} \pb{x_t|x_{\tau},q_{[\tau]}} = \sum_{x_t} \max_{x_{\tau}} \pb{x_t|x_{\tau}},\]
where the last step follows because of the Markov chain
$Q_{[\tau]} \rightarrow X_{\tau} \rightarrow X_t$, 
as in \eqref{eq:appendix-A-Markov}.
 
This completes the proof.

\section{Proof of Lemma~\ref{lemma:lower}}
\label{Appendix-lemma-lower}
Consider 
    \begin{align*}
        \max_{u \in \cN} p\left(x|u\right) 
        & \utag{a}{=} \max_{u \in \cN} \sum_{y:x \in y} p\left(x, y|u\right) \\
        & = \max_{u \in \cN} \sum_{y:x \in y} p\left(y|u\right) p\left(x|y,u\right) \\
        & \utag{b}{=} \max_{u \in \cN} \sum_{y:x \in y} p\left(y\right) p\left(x| y,u\right) \\
        & \leq \sum_{y:x \in y} p\left(y\right) \max_{u \in \cN}  p\left(x| y,u\right) \\
        & \leq \sum_{y:x \in y} p\left(y\right), 
    \end{align*}
where \uref{a} follows from $p(x,y)=0$ for $x \notin y$, and \uref{b} follows because $Y$ is independent of $U$.

Thus, we obtain that 
\begin{align*}
     \sum_{x \in \cN} \max_{u \in \cN} p\left(x|u\right) 
    & \leq  \sum_{x \in \cN} \sum_{y:x \in y} p\left(y\right) \\
    & =  \sum_{y \in \sP\left(\cX\right)} \sum_{x:x \in y} p\left(y\right) \\
    & =  \sum_{y \in \sP\left(\cN\right)} p\left(y\right) \sum_{x:x \in y} 1 \\
    & =  \sum_{y \in \sP\left(\cN\right)} p(y) |y| \\
    & = \mathbb{E}\left[|Y|\right],
\end{align*}
which completes the proof.

\section{Justification of the algorithm for Lemma~\ref{lemma:algorithm}}
\label{app:proof-lemma-achievability}

\subsection{Verification of \eqref{eq:proof-existence}}
First, let us verify \eqref{eq:proof-existence}, \Ie
for any $\ell=1,\ldots,\sigma+1$ and $x \in \cN$,
\begin{equation*}
        \sum_{k=1}^{N} Q_{u^{(x,i)},k} \geq \min\left\{\delta_{x},  \pb{X=x|U=u^{(x,\ell)}}\right\}- \pb{X=x|U=u^{(x,\ell-1)}}, \forall i =1,\ldots,\ell-1.
\end{equation*}
Roughly speaking, the summation of the $u^{(x,i)}$-th row of $Q$ should be larger than or equal to the right-hand side of \eqref{eq:proof-existence} at any points when the algorithm update. Equivalently, for any given $\ell=1,\ldots,\sigma+1$ and $x,u \in \cN$, if $u \in \left\{u^{(x,i)}: i =1,\ldots,\ell-1 \right\}$, we need to verify that  
\begin{equation}
\label{eq:justification-t1}
    \sum_{k=1}^{N} Q_{u,k} \geq \min\left\{\delta_{x},  \pb{X=x|U=u^{(x,\ell)}}\right\}- \pb{X=x|U=u^{(x,\ell-1)}},
\end{equation}

From \eqref{eq:alg-sum} and \eqref{eq:alg-t3}, it is clear that we subtract exactly the same value as the right-hand side of \eqref{eq:proof-existence} from $\sum_{k=1}^{N} Q_{u,k}$ during each update. Therefore, by summing over $x$ and $\ell$, it is sufficient to show that for any given $u$, we have
\begin{equation*}
    \sum_{k=1}^{N} Q_{u,k} \geq \sum_{\ell=1}^{\sigma+1} \sum_{x \in \cN: u \in \left\{u^{(x,i)}: i =1,\ldots,\ell-1 \right\} }  \min\left\{\delta_{x},  \pb{X=x|U=u^{(x,\ell)}}\right\}- \pb{X=x|U=u^{(x,\ell-1)}},
\end{equation*}
where $Q_{u,k}$ denotes the initializations in \eqref{eq:initialize}. To be precise, we re-write it by 
  \begin{equation}
\label{eq:justification-t2}
\begin{aligned}
& \sum_{k=1}^{N} \max\left\{\pb{X=k|U=u}-\delta_k, 0\right\} \\
& \geq \sum_{\ell=1}^{\sigma+1} \sum_{x \in \cN: u \in \left\{u^{(x,i)}: i =1,\ldots,\ell-1 \right\} }  \min\left\{\delta_{x},  \pb{X=x|U=u^{(x,\ell)}}\right\}- \pb{X=x|U=u^{(x,\ell-1)}}.
\end{aligned}
\end{equation}

To establish \eqref{eq:justification-t2},
for a given $u \in \cN$, let us suppose that 
\begin{equation}
    u=u^{(1,\alpha_1)}=\cdots=u^{(N,\alpha_N)}.
\end{equation}
Then, the left-hand side of \eqref{eq:justification-t2} can be written as 
\[ \sum_{k: \alpha_k \geq \sigma+1} \left(\pb{X=k|U=u}-\delta_k\right),\]
while the right-hand side of \eqref{eq:justification-t2} can be written as
\begin{align*} 
    & \sum_{x \in \cN} \sum_{\ell: u \in \left\{u^{(x,i)}: i =1,\ldots,\ell-1 \right\}}   \min\left\{\delta_{x},  \pb{X=x|U=u^{(x,\ell)}}\right\}- \pb{X=x|U=u^{(x,\ell-1)}}\\
    & = \sum_{x: \alpha_x \leq \sigma} 
    \sum_{\ell = \alpha_{x}+1}^{\sigma+1}   
    \min\left\{\delta_{x},  \pb{X=x|U=u^{(x,\ell)}}\right\}- \pb{X=x|U=u^{(x,\ell-1)}} \\
    & = \sum_{x: \alpha_x \leq \sigma} 
    \sum_{\ell = \alpha_{x}+1}^{\sigma+1}  
    \left(\delta_{x}-  \pb{X=x|U=u^{(x,\alpha_x)}}\right) \\
    & = \sum_{x: \alpha_x \leq \sigma} 
    \left(\delta_{x}-  \pb{X=x|U=u}\right).
\end{align*} 
Therefore, it remains to show that 
\[\sum_{k: \alpha_k \geq \sigma+1} \left(\pb{X=k|U=u}-\delta_k\right) \geq \sum_{k \in \cN: \alpha_k \leq \sigma} 
    \left(\delta_{k}-  \pb{X=k|U=u}\right), 
 \]
which can be written as 
\begin{equation}
\label{eq:justification-zero}
\sum_{k \in \cN} \left(\pb{X=k|U=u}-\delta_k\right) \geq 0. 
\end{equation}
Since 
\[\sum_{k \in \cN} \pb{X=k|U=u} = \sum_{k \in \cN} \delta_k =1, 
\]
we can easily see that \eqref{eq:justification-zero} holds, which completes the proof.

One may notice that we indeed show that the equality holds in \eqref{eq:justification-t2}, which implies that $Q$ would be an all-zeros matrix after iterations, \Ie for any given $\bar{x},\bar{u} \in \cN$,   
\begin{equation}
\label{eq:proof-all-zero}
    \sum_{\ell=1}^{\sigma+1} \sum_{x=1}^{N} 
    \sum_{i,j:x_{i,j}=\bar{x},u^{(x,i)}=\bar{u}}  v_{i,j} = 
    \max\left\{\pb{X=\bar{x}|U=\bar{u}}-\delta_{\bar{x}}, 0\right\}.
\end{equation}
Here, we slightly abuse the notation since $v_{i,j}$ should be independent of $x$ and $\ell$ as described.

\subsection{Justification of the algorithm}
In this subsection, we will verify that the proposed algorithm works, \Ie it ends up with producing a distribution $\pb{z,x|u}$ satisfying that 
\begin{equation}
\label{eq:proof-decode-jusify}
    \pb{z,x}=0,~\forall x \notin z,
\end{equation}
\begin{equation}
\label{eq:proof-privacy-jusify}
    \pb{z|u}=\pb{z|u'},~\forall z \in \cZ~\text{and}~u,u' \in \cN,
\end{equation}
and 
\begin{equation}
\label{eq:proof-rate-jusify}
    \pb{|Z| = i} = \theta_i,~\forall i=1,\ldots,\sigma+1.
\end{equation}

As claimed, $\qb{z,x,u}$ stores the non-zero valued probability of $\pb{z,x|u}$. To establish this claim, we need to verify that 
\begin{equation}
\label{eq:proof-distribution}
    \sum_{z':(z',x,u) \in \cA} \qb{z',x,u} = \pb{x|u},~\forall x,u \in \cN.  
\end{equation}

Since $\cA$ is the union set of $\cA_{k,x,\ell}$ for all possible $k$, $x$ and $\ell$, let us focus on $\cA_{k,x,\ell}$ defined in \eqref{eq:proof-arguments}. Recall that  
\begin{equation*}
\begin{aligned}
    \cA_{k,x,\ell} = & \left\{\left(\bar{z},\bar{x},\bar{u}\right): \bar{z}=z_k, \bar{x}=\zeta_k(i), \bar{u}=u^{(x,i)}, i=1,\ldots,\ell-1 \right\} \\
    & ~~~~~ \bigcup \left\{\left(\bar{z},\bar{x},\bar{u}\right):\bar{z}=z_k,\bar{x}=x,\bar{u} \in \cN\setminus \left\{u^{(x,i)}:i=1,\ldots,\ell-1 \right\} \right\}.
\end{aligned}
\end{equation*}
Denote 
\[\cA^{(1)}_{k,x,\ell} = \left\{\left(\bar{z},\bar{x},\bar{u}\right): \bar{z}=z_k, \bar{x}=\zeta_k(i), \bar{u}=u^{(x,i)}, i=1,\ldots,\ell-1 \right\},\]
and 
\[\cA^{(2)}_{k,x,\ell} = \left\{\left(\bar{z},\bar{x},\bar{u}\right):\bar{z}=z_k,\bar{x}=x,\bar{u} \in \cN\setminus \left\{u^{(x,i)}:i=1,\ldots,\ell-1 \right\} \right\}.\]

For any given $\bar{u}, \bar{x} \in \cN$, we have
\begin{align}
    \sum_{z:\left(z,\bar{u}, \bar{x}\right) \in \cA} \qb{z,\bar{x},\bar{u}} 
    & \utag{a}{=} \sum_{\ell=1}^{\sigma+1} \sum_{z:\left(z,\bar{u}, \bar{x}\right) \in \cA_{\ell}} \qb{z,\bar{x},\bar{u}}  \nonumber \\
    & \utag{b}{=} \sum_{\ell=1}^{\sigma+1} 
    \sum_{x=1}^{N}  \sum_{k:\left(\cdot,\bar{x},\bar{u}\right) \in \cA_{k,x,\ell}}  \nu_{k,x,\ell} \nonumber \\
     & = \sum_{\ell=1}^{\sigma+1} 
    \sum_{x=1}^{N} 
    \left(\sum_{k:\left(\cdot,\bar{x},\bar{u}\right) \in \cA^{(1)}_{k,x,\ell}}  \nu_{k,x,\ell} + \sum_{k:\left(\cdot,\bar{x},\bar{u}\right) \in \cA^{(2)}_{k,x,\ell}}  \nu_{k,x,\ell} \right) \nonumber \\
    & = \sum_{\ell=1}^{\sigma+1} 
    \sum_{x=1}^{N}  \sum_{k:\left(\cdot,\bar{x},\bar{u}\right) \in \cA^{(1)}_{k,x,\ell}}  \nu_{k,x,\ell} + 
    \sum_{\ell=1}^{\sigma+1} 
    \sum_{x=1}^{N} \sum_{k:\left(\cdot,\bar{x},\bar{u}\right) \in \cA^{(2)}_{k,x,\ell}}  \nu_{k,x,\ell}. \label{eq:proof-distribution-two-terms}
\end{align}
where \uref{a} follows because $\cA_{\ell}$ are disjoint for distinct $\ell$ and \uref{b} follows from \eqref{eq:assign}.

For any given $\bar{u} \in \cN$, suppose that 
\begin{equation}
    \bar{u} = u^{(1,\alpha_1)} = \cdots = u^{(N,\alpha_N)}. 
\end{equation}
Then, the first term of the right-hand side of \eqref{eq:proof-distribution-two-terms} can be written as 
\begin{align}
    \sum_{\ell=1}^{\sigma+1} \sum_{x=1}^{N} \sum_{k:\left(\cdot,\bar{x},\bar{u}\right) \in \cA^{(1)}_{k,x,\ell}}  \nu_{k,x,\ell}
    & = \sum_{\ell=1}^{\sigma+1} \sum_{x:\alpha_x \leq \ell-1}   \sum_{k: \bar{x}=\zeta_k(\alpha_x)}  \nu_{k,x,\ell} \nonumber \\
    & \utag{a}{=} \sum_{\ell=1}^{\sigma+1} \sum_{x:\alpha_x \leq \ell-1}   v_{\alpha_x,\bar{x}} \nonumber \\
    & \utag{b}{=} \max\left\{\pb{X=\bar{x}|U=\bar{u}}-\delta_{\bar{x}}, 0 \right\}, \label{eq:proof-term1}
\end{align}
where \uref{a} follows from \eqref{eq:proof-multiset-value}, and \uref{b} follows from \eqref{eq:proof-all-zero}. Note that we slightly abuse the notation $\zeta_k$ and $v_{\alpha_k,\bar{x}}$ here since they are independent of $x$ and $\ell$. 

The second term of the right-hand side of \eqref{eq:proof-distribution-two-terms} can be written as 
\begin{align}
    \sum_{\ell=1}^{\sigma+1} \sum_{x=1}^{N} \sum_{k:\left(\cdot,\bar{x},\bar{u}\right) \in \cA^{(2)}_{k,x,\ell}}  \nu_{k,x,\ell}
    & = \sum_{\ell=1}^{\sigma+1} \sum_{x:x=\bar{x}} \mathbbm{1}_{\alpha_x \geq \ell}  \sum_{k}  \nu_{k,x,\ell} \nonumber \\
    & = \sum_{\ell=1}^{\sigma+1} \mathbbm{1}_{\alpha_{\bar{x}} \geq \ell} \sum_{k}  \nu_{k,\bar{x},\ell} \nonumber \\
    & \utag{a}{=} \sum_{\ell=1}^{\sigma+1} \mathbbm{1}_{\alpha_{\bar{x}} \geq \ell} \left(  \min\left\{\delta_{\bar{x}},  \pb{X=\bar{x}|U=u^{(\bar{x},\ell)}}\right\}- \pb{X=\bar{x}|U=u^{(\bar{x},\ell-1)}} \right) \nonumber \\
    & = \sum_{\ell=1}^{\min\{\sigma+1,\alpha_{\bar{x}}\}} \left(  \min\left\{\delta_{\bar{x}},  \pb{X=\bar{x}|U=u^{(\bar{x},\ell)}}\right\}- \pb{X=\bar{x}|U=u^{(\bar{x},\ell-1)}} \right) \nonumber \\
    & =  \min\left\{\delta_{\bar{x}},\pb{X=\bar{x}|U=u^{(\bar{x},\alpha_{\bar{x}})}}\right\}  \nonumber \\
    & = \min\left\{\delta_{\bar{x}},\pb{X=\bar{x}|U=\bar{u}}\right\}, \label{eq:proof-term2}
\end{align}
where \uref{a} follows from \eqref{eq:alg-sum} and \eqref{eq:proof-multiset-value}.
 
Finally, it is easy to see that 
\begin{align*}
    \sum_{\ell=1}^{\sigma+1} \sum_{x=1}^{N} \sum_{k:\left(\cdot,\bar{x},\bar{u}\right) \in \cA_{k,x,\ell}}  \nu_{k,x,\ell} & = \sum_{\ell=1}^{\sigma+1} \sum_{x=1}^{N} \sum_{k:\left(\cdot,\bar{x},\bar{u}\right) \in \cA^{(1)}_{k,x,\ell}}  \nu_{k,x,\ell} + \sum_{\ell=1}^{\sigma+1} \sum_{x=1}^{N} \sum_{k:\left(\cdot,\bar{x},\bar{u}\right) \in \cA^{(2)}_{k,x,\ell}}  \nu_{k,x,\ell} \\
    & \utag{a}{=} \max\left\{\pb{X=\bar{x}|U=\bar{u}}-\delta_{\bar{x}}, 0 \right\} + \min\left\{\delta_{\bar{x}},\pb{X=\bar{x}|U=\bar{u}}\right\} \\
    & = \pb{X=\bar{x}|U=\bar{u}},
\end{align*}
where \uref{a} follows from \eqref{eq:proof-term1} and \eqref{eq:proof-term2}. We finish justifying \eqref{eq:proof-distribution}.


Now, let us verify the two constraints \eqref{eq:proof-decode-jusify} and \eqref{eq:proof-privacy-jusify}, \Ie 
\begin{equation*}
    \pb{z,x}=0,~\forall x \notin z,
\end{equation*}
and
\begin{equation*}
    \pb{z|u}=\pb{z|u'},~\forall z \in \cZ~\text{and}~u,u' \in \cN.
\end{equation*}

As we have shown that 
\begin{equation*}
\pb{z,x|u} = 
\begin{cases}
    \qb{z,x,u}, & (z,x,u) \in \cA, \\
    0, & (z,x,u) \notin \cA,
\end{cases}
\end{equation*}
to verify the two constraints, it is equivalent to show that
\begin{enumerate}
     \item For any $(z,x,u) \in \cA$, it must have $x \in z$.
     \item For any given $z \in \cZ$ and $u,u' \in \cN$, we have  
    \begin{equation*}
    \sum_{x:(z,x,u) \in \cA}\qb{z,x,u} = \sum_{x:(z,x,u') \in \cA}\qb{z,x,u'}.
    \end{equation*}
 \end{enumerate}

Since $\cA$ is the union set of $\cA_{k,x,\ell}$ for all possible $k$, $x$ and $\ell$, it is sufficient to show the following two claims:
\begin{enumerate}
     \item For any $(z,x,u) \in \cA_{k,x,\ell}$, it must have $x \in z$.
     \item For any given $z \in \cZ$ and $u,u' \in \cN$, we have  
 \begin{equation}
    \label{eq:proof-justify-two-constraint-2} 
     \sum_{x,\bar{x},k:\left(z,x,u\right) \in \cA_{k,\bar{x},\ell}}  \nu_{k,\bar{x},\ell} =  \sum_{x,\bar{x},k:\left(z,x,u'\right) \in \cA_{k,\bar{x},\ell}}  \nu_{k,\bar{x},\ell},
    \end{equation}
    where $\ell=|z|$, $x, \bar{x} = 1,\ldots,N$ and $k=1,\ldots,e_{\bar{x},\ell}$.
 \end{enumerate} 


Recall the definition of $\cA_{k,x,\ell}$ for any $k$, $x$ and $\ell$, \Ie 
\begin{equation*}
\begin{aligned}
    \cA_{k,x,\ell} = & \left\{\left(\bar{z},\bar{x},\bar{u}\right): \bar{z}=z_k, \bar{x}=\zeta_k(i), \bar{u}=u^{(x,i)}, i=1,\ldots,\ell-1 \right\} \\
    & ~~~~~ \bigcup \left\{\left(\bar{z},\bar{x},\bar{u}\right):\bar{z}=z_k,\bar{x}=x,\bar{u} \in \cN\setminus \left\{u^{(x,i)}:i=1,\ldots,\ell-1 \right\} \right\}.
\end{aligned}
\end{equation*}

Since $z_k=\{\zeta_k,x\}$ as previously defined, we can easily see that $\bar{x} \in \bar{z}$ for any $(\bar{z},\bar{x},\bar{u}) \in \cA_{k,x,\ell}$, which justifies the first claim.

For the second claim, we re-write \eqref{eq:proof-justify-two-constraint-2} by
\begin{equation*}
\sum_{\bar{x},k} \nu_{k,\bar{x},\ell} \left(\sum_{x:\left(z,x,u\right) \in \cA_{k,\bar{x},\ell}} \right) =
\sum_{\bar{x},k} \nu_{k,\bar{x},\ell} \left(\sum_{x:\left(z,x,u'\right) \in \cA_{k,\bar{x},\ell}} \right).
\end{equation*} 
By inspecting the definition of $\cA_{k,x,\ell}$, we can see that there exists exactly one tuple $(z,\cdot,u) \in \cA_{k,x,\ell}$ for any given $u$ and $z$, so we have
\begin{equation*}
\sum_{x:\left(z,x,u\right) \in \cA_{k,\bar{x},\ell}}  =
\sum_{x:\left(z,x,u'\right) \in \cA_{k,\bar{x},\ell}}, 
\end{equation*} 
which completes proving \eqref{eq:proof-justify-two-constraint-2}.

Finally, let us justify \eqref{eq:proof-rate-jusify}, \Ie
\begin{equation*}
    \pb{|Z| = i} = \theta_i,~\forall i=1,\ldots,\sigma+1,
\end{equation*}
whose proof is given as follows:
\begin{align*}
 \pb{|Z| = \ell} 
 & = \sum_{z:|z|=\ell} \sum_{u}\pb{u} \pb{z|u} \\ 
 & \utag{a}{=}  \sum_{z:|z|=\ell} \pb{z|\bar{u}} \\ 
 & = \sum_{z:|z|=\ell} \sum_{x} \qb{z,x,\bar{u}} \\
 & = \sum_{(z,x,\bar{u}):(z,x,\bar{u}) \in \cA_{\ell}}  \qb{z,x,\bar{u}} \\
& \utag{b}{=}  \sum_{(z,x,\bar{u}) \in \cA_{\ell}} \sum_{x'=1}^{N} 
         \sum_{k:\left(z,x,\bar{u}\right) \in \cA_{k,x',\ell}}  \nu_{k,x',\ell} \\
& =  \sum_{x'=1}^{N} \sum_{k} \sum_{(z,x,\bar{u}):\left(z,x,\bar{u}\right) \in \cA_{k,x',\ell}}  \nu_{k,x',\ell} \\
& \utag{c}{=}  \sum_{x'=1}^{N} \sum_{k}  \nu_{k,x',\ell} \\
& \utag{d}{=}  \sum_{x'=1}^{N} \min\left\{\delta_{x'},  \pb{X=x'|U=u^{(x',\ell)}}\right\}- \pb{X=x'|U=u^{(x',\ell-1)}} \\ 
    & =  \theta_\ell,         
\end{align*}
where \uref{a} follows from \eqref{eq:proof-privacy-jusify}, \uref{b} follows from \eqref{eq:assign}, \uref{c} follows because there exists exactly one tuple $(\cdot,\cdot,\bar{u}) \in \cA_{k,x,\ell}$ for given $k$, $x$ and $\ell$, and \uref{d} follows from \eqref{eq:alg-sum} and \eqref{eq:proof-multiset-value}.

\subsection{Complexity analysis of the algorithm}
In this subsection, we will discuss the complexity of the algorithm to construct the desired output distribution $\pb{z|x,u}$. The purpose of the complexity analysis here is to justify that the proposed algorithm is tractable, \Ie with $\text{poly}(N)$ complexity. By utilizing some data structures, one may possibly reduce the complexity by one or two orders, which is beyond the interest of this paper.

\emph{Pre-calculation:} The Pre-calculation involves two steps, \Ie sorting $\pb{x|u}$ for all $x \in \cN$ and picking the set $\{\delta_j:j=1,\ldots,N\}$. The complexity of sorting is $\cO\left(N^2\log N\right)$ and picking $\{\delta_j:j=1,\ldots,N\}$ is $\cO\left(N\right)$. 

\emph{Initialization:} The initialization of $Q$ is $\cO(N^2)$.



\emph{Procedure:} The main procedure is divided into the following steps:
\begin{enumerate}
    \item For the fixed $\ell$, $x$ and $u_i$, we `randomly' choose a collection of pairs $I_i \times V_i$. We can easily see that if \eqref{eq:proof-existence} is satisfied, then $V_i$ (and $I_i$) can be chosen by linear time $\cO(N)$, \Ie going through the $u_i$-th row of the matrix $Q$. Hence, we can obtain $\{I_i,V_i:i=1,\ldots,\ell-1\}$ for a fixed $\ell$ and $x$ with $\cO((\ell-1)N)$.
    \item For a fixed $\ell$ and $x$,  we need to get a collection of pairs $\left\{( \zeta_{k},\nu_{k}):k=1,2,\ldots,e \right\}$ given $\{I_i,V_i:i=1,\ldots,\ell-1\}$. Each $\nu_k$ is obtained by finding the minimal value of $\mathrm{B}_v$, which is a set of length $\ell-1$, so finding each $\nu_k$ (and $\zeta_k$) takes $\cO(\ell-1)$. For each $z_k$, the set $\cA_k$ can be characterized by traversing $z_k$ with linear time $\cO(\ell-1)$. As each $e_i$ is bounded by $N-1$, $e$ is bounded by 
    \[e \leq \sum_{i=1}^{\ell-1} e_i = (\ell-1)(N-1),\]
    and hence determining all $\left\{ \cA_{k},\nu_{k}:k=1,2,\ldots,e \right\}$ takes $\cO((\ell-1)^2(N-1))$. Therefore, obtaining $\{\cA_{k,x,\ell}, \nu_{k,x,\ell}: 1 \leq \ell \leq \sigma+1, 1 \leq x \leq N, 1 \leq k \leq e_{x,\ell}
\}$ at most takes $\cO(\sigma^3 N^2)$.
    \item At the end, we need to finish the probability assignment (c.f.\eqref{eq:assign}). However, since the size of the alphabet of $Z$ is exponential, $\pb{z,x|u}$ has an exponential number of elements. To avoid the exponential overhead, we may take advantage of the sparsity of $\pb{z,x|u}$ to output non-zero positions and values (all others are assumed to be zero) instead of pushing out the distribution $\pb{z,x|u}$ entirely and directly. Indeed, $\cA_{k,x,\ell}$ contains the non-zero positions and the corresponding value is $\nu_{k,x,\ell}$. However, since some positions may appear in $\cup_{k,x,\ell} \cA_{k,x,\ell}$ multiple times, we may need to merge them, this can be done by simply checking all $\{\cA_{k,x,\ell}: 1 \leq \ell \leq \sigma+1, 1 \leq x \leq N, 1 \leq k \leq e_{x,\ell}\}$ which is $\cO(\sigma^3 N^3)$.   
\end{enumerate}

In summary, the worst case complexity of the algorithm is $\cO(N^6)$.

\end{document}

%% file: preamble.tex
\usepackage{graphicx}                                      
\usepackage{amssymb,amsthm,bm}  
\usepackage[cmex10]{amsmath} 
\usepackage{bbm}

\usepackage{multirow}   
  
\usepackage{algorithm}
\usepackage{algorithmic}  

\usepackage{diagbox}
\usepackage{color}
\usepackage{xcolor}

\usepackage{tikz} 
\usetikzlibrary{automata,positioning,chains,fit,shapes,calc}
\usetikzlibrary{arrows}
\usepackage{tkz-graph}
\usepackage{caption}
\usepackage{booktabs}
\usepackage{cite}

\makeatletter
\newcommand{\gettikzxy}[3]{%
  \tikz@scan@one@point\pgfutil@firstofone#1\relax
  \edef#2{\the\pgf@x}%
  \edef#3{\the\pgf@y}%
}
\makeatother

\makeatletter
\newcommand{\mathleft}{\@fleqntrue\@mathmargin0pt}
\newcommand{\mathcenter}{\@fleqnfalse}
\makeatother              
\newtheorem{theorem}{Theorem}
\newtheorem{lemma}{Lemma}
\newtheorem{corollary}[lemma]{Corollary}

\newtheorem{definition}{Definition}

\newtheorem{proposition}{Proposition}

\newtheorem{example}{Example}
\newcommand{\utag}[2]{\mathop{#2}\limits^{\text{(#1)}}}
\newcommand{\uref}[1]{(#1)}

\newcommand{\Ie}{\textit{i.e., }}
\newcommand{\Eg}{\textit{e.g., }}

\long\def\symbolfootnote[#1]#2{\begingroup
\def\thefootnote{\fnsymbol{footnote}}\footnote[#1]{#2}\endgroup}

\newcommand{\cA}{\mathcal{A}}
\newcommand{\cB}{\mathcal{B}}

\newcommand{\cD}{\mathcal{D}}

\newcommand{\cF}{\mathcal{F}}

\newcommand{\cN}{\mathcal{N}}
\newcommand{\cO}{\mathcal{O}}

\newcommand{\cQ}{\mathcal{Q}}

\newcommand{\cS}{\mathcal{S}}

\newcommand{\cX}{\mathcal{X}}

\newcommand{\cZ}{\mathcal{Z}}
\usepackage{mathrsfs}

\newcommand{\sF}{\mathscr{F}}

\newcommand{\sP}{\mathscr{P}}

\newcommand{\sW}{\mathscr{W}}
\newcommand{\sX}{\mathscr{X}}

\usepackage{xcolor}
\usepackage{tikz}
\usetikzlibrary{automata,positioning,chains,fit,shapes,calc,arrows,plotmarks,arrows.meta}
\usepackage{tkz-graph}
\usepackage{caption}
\usepackage{booktabs}
\usepackage{cite}
\usepackage{hyperref}



\usepackage{pgfplots}
  \pgfplotsset{compat=newest}

  \usepgfplotslibrary{patchplots}
  \usepackage{grffile}
  
  \usepackage{ dsfont }
  
\definecolor{light_gray}{RGB}{220,220,220}

\makeatletter
\def\thickhline{%
  \noalign{\ifnum0=`}\fi\hrule \@height \thickarrayrulewidth \futurelet
   \reserved@a\@xthickhline}
\def\@xthickhline{\ifx\reserved@a\thickhline
               \vskip\doublerulesep
               \vskip-\thickarrayrulewidth
             \fi
      \ifnum0=`{\fi}}
\makeatother

\newlength{\thickarrayrulewidth}
\definecolor{lightblue}{RGB}{220, 235, 242}

\usepackage{subcaption}

%% file: preamble_figure.tex
\usepackage{tikz}
\usetikzlibrary{er,automata,positioning,chains,fit,shapes,calc,arrows,plotmarks,arrows.meta}
\usepackage{tkz-graph}
\usepackage{caption}
\usepackage{varwidth}
\usetikzlibrary{calc}

\definecolor{darkgreen}{RGB}{34,139,34}
\definecolor{bleudefrance}{rgb}{0.19, 0.55, 0.91}

\definecolor{darkgreen2}{RGB}{242,248,241}
\definecolor{bleudefrance2}{RGB}{229, 240, 253}
\definecolor{magenta2}{RGB}{252, 241, 249}

\usepackage{transparent}

\newcommand{\Markov}[3]{
\begin{scope}[shift={#1},scale=#3]
\tikzstyle{every node}=[font=\scriptsize]
        \tikzset{node style/.style={state, 
        							inner sep=0.5pt,
                                    minimum size = 12pt,
                                    line width=0.05mm,
                                    fill=white},
                    LabelStyle/.style = { minimum width = 1em,
                                            text = black},
                   EdgeStyle/.append style = {->, bend left=22, line width=0.05mm} }

        \node[node style] at (4,7)     (pro-left)     {1};
        \node[node style] at (5.75, 7)     (pro-right)     {2};
        \node[node style] at (5.25, 5.5)     (N)     {N};
		  \Edge(pro-left)(pro-right);
  			\Edge(pro-right)(pro-left);
  			\Edge(pro-left)(N);
  			\Edge(N)(pro-left);
            
            \path[<-] (pro-left) edge  [loop left] node {} ();
            \path[<-] (pro-right) edge  [loop right] node {} ();
\end{scope}
}

%% file: setting_figure.tex
\begin{tikzpicture}[auto,node distance=1.5cm]


   \node at (3.75,4.65) {Server};  
  \node at (3.8,2.5) [rectangle,rounded corners=5mm,minimum height=3.5cm,minimum width=4.5cm, draw] (server) {};
  \node at (2.25,3.25) [rectangle,draw] {$\sW_1$};
  \node at (3.25,3.25) [rectangle,draw] {$\sW_2$};
  \node at (4.25,3.25) {$\dots$};
  \node at (5.25,3.25) [rectangle,draw] {$\sW_N$};
  
\draw [decorate,decoration={brace,amplitude=7pt},xshift=0pt,yshift=0pt, color=blue, line width=0.5pt]
(1.8,3.6) -- (5.8,3.6) node [black,midway,xshift=0cm, yshift=0.15cm] 
{\scriptsize \color{blue} Information Sources};

\node at (2.1,2.37) [ rotate=90] {\scriptsize generates};

\draw[->]  (2.25,2.95) -- (2.25,1.72) node {};
\node at (2.25,1.45) {$W_{1,t}$};

\draw[->]  (3.25,2.95) -- (3.25,1.72) node {};
\node at (3.25,1.45) {$W_{2,t}$};

\draw[->]  (5.25,2.95) -- (5.25,1.72) node {};
\node at (5.25,1.45) {$W_{N,t}$};

 \draw [decorate,decoration={brace,amplitude=7pt},xshift=0pt,yshift=0pt, color=violet, line width=0.5pt]
(5.8,1.35) -- (1.8,1.35) node [black,midway,xshift=0cm, yshift=-0.1cm] 
{\scriptsize \color{violet} Messages at time $t$};


 \node at (11.85,4.65) {User};  
  \node at (11.85,2.5) [rectangle,rounded corners=5mm,minimum height=3.5cm,minimum width=5cm, draw] (server) {};
\Markov{(10.5,0.15)}{}{0.5}
  \node at (13.35,3.27) [ rotate=65] {\tiny $\dots$};
  
\node at (12.95,3.25) [rectangle,draw,minimum width=2.35cm, minimum height=1.4cm] {};
 \draw[->] (11.8,3.25) -- (10.5,3.25){};
 \node at (11.16,3.4) {\scriptsize generates};
  \node at (10,3.25) {$X_t$};
  \node at (10,2.95) {\scriptsize\color{red}User's};
  \node at (10,2.65) {\scriptsize\color{red}Request};
  
  \node at (12.15,1.7) [rectangle,draw,minimum width=0.75cm, minimum height=0.75cm] {$\sF$};
 \draw[->] (11.8,1.7) -- (10.5,1.7){};
 \node at (11.16,1.85) {\scriptsize generates};
  \node at (10,1.7) {$F_t$};
  \node at (10,1.35) {\scriptsize\color{darkgreen}Privacy};
  \node at (10,1.1) {\scriptsize\color{darkgreen}Mode};
  
   \node at (11.17,1.5) {\scriptsize\color{darkgreen} ON/OFF};
   
  
  
 \node at (13.4,1.45){ \includegraphics[scale=0.35]{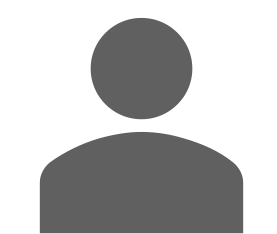}};
  
  
  \draw[<-,thick]  (6.35,3.2) -- (9,3.2) node {};
  \draw[->,thick]  (6.35,1.7) -- (9,1.7) node {};
  
  \node at (7.65,2.9  ) {$Q_t$};
  \node at (7.65,1.45) {$A_t$};
    \node at (7.65,3.45) {Query};
  \node at (7.65,1.95) {Answer};

\end{tikzpicture}

%% file: symmetric_markov.tex
\begin{tikzpicture}[thick,scale=1, every node/.style={scale=1}]
 \tikzstyle{every node}=[font=\normalsize]
        \tikzset{node style/.style={state, 
        							inner sep=0.5pt,
                                    minimum size = 25pt,
                                    line width=0.2mm,
                                    fill=white},
                    LabelStyle/.style = { minimum width = 1em, fill = white!10,
                                            text = black},
                   EdgeStyle/.append style = {->, bend left=18, line width=0.2mm} }

        \node at (6.5,9) [rectangle,minimum width=1cm, minimum height=4.8cm] {};
        
        \node at (6.5,8.6) [line width=0.1pt, rectangle,draw,minimum width=6.3cm, minimum height=4.55cm,color=gray] {};

        \node[node style] at (5,10)     (1)     {1};
        \node[node style] at (8, 10)     (2)     {2};
		\node[node style] at (6.5, 8)     (3)     {3};

            \Edge[label=$\frac{1-\alpha}{2}$](1)(2)
  			\Edge (2)(1)
  			\Edge (1)(3)
  			\Edge[label=$\frac{1-\alpha}{2}$](3)(1)
  			\Edge (3)(2)
  			\Edge[label=$\frac{1-\alpha}{2}$](2)(3)
            
            \Loop[dist = 1.5cm, dir = NO, label = $\alpha$](1.west)
            \Loop[dist = 1.5cm, dir = SO, label = $\alpha$](2.east)           
            \Loop[dist = 1.5cm, dir = WE, label = $\alpha$](3.south)     
            
            \node at (6.5,6.8) [rectangle,minimum width=1cm, minimum height=3.7cm] {};
         
    \end{tikzpicture}

%% file: symmetric_plot.tex
\definecolor{mycolor1}{rgb}{0.00000,0.44700,0.74100}%
\definecolor{mycolor2}{rgb}{0.85000,0.32500,0.09800}%
\definecolor{mycolor3}{rgb}{0.92900,0.69400,0.12500}%
\definecolor{mycolor4}{rgb}{0.49400,0.18400,0.55600}%
\definecolor{darkred}{rgb}{0.64, 0.0, 0.0}

\begin{tikzpicture}
\begin{axis}[width=5cm,
height=4.5cm,
scale only axis,
xmin=0,
xmax=1,
ymin=0.3,
ymax=1,
xlabel=$\alpha$,
axis background/.style={fill=white},
legend style={legend cell align=left, align=left, draw=none}
]

\addplot[smooth,mark=none,line width=1.5pt,domain=0.001:0.33,
samples=5,
color=mycolor1,forget plot, mark=triangle*, mark options={solid, rotate=270, fill=mycolor1, mycolor1}]  {1/(2-3*x)};

\addplot[smooth,mark=none,line width=1.5pt,domain=0.33:0.999,
samples=10,
color=mycolor1, mark=triangle*, mark options={solid, rotate=270, fill=mycolor1, mycolor1}]  {1/(3*x)};
\addlegendentry{$R_t^I$}

\addplot[color=red, line width=1.5pt,domain=0.001:0.33,
samples=60, forget plot]  {2/(3-3*x)};

\addplot[color=red, line width=1.5pt,,domain=0.33:0.999,
samples=60]  {1/(3*x)};
\addlegendentry{$R_t^O$}

\end{axis}

\end{tikzpicture}

%% file: plot.tex
%
%
\definecolor{mycolor1}{rgb}{0.00000,0.44700,0.74100}%
\definecolor{mycolor2}{rgb}{0.85000,0.32500,0.09800}%
\definecolor{mycolor3}{rgb}{0.92900,0.69400,0.12500}%
\definecolor{mycolor4}{rgb}{0.49400,0.18400,0.55600}%
\definecolor{darkred}{rgb}{0.64, 0.0, 0.0}
\begin{tikzpicture}

\begin{axis}[%
width=2.5in,
height=2in,
at={(1.011in,0.642in)},
scale only axis,
xmin=0,
xmax=20,
xtick={0, 2, 4, 6, 8, 10, 12, 14, 16, 18, 20},
xlabel style={font=\color{white!15!black}},
xlabel={ $t-\tau$},
ymin=0.5,
ymax=1,
ylabel style={font=\color{white!15!black}},
ylabel={Download Rate $R_t$},
axis background/.style={fill=white},
legend style={at={(0.755,0.345)}, anchor=south west, legend cell align=left, align=left, draw=white!15!black}
]

 \addlegendimage{empty legend}
 \addlegendentry{\hspace{-.55cm}\textbf{$\alpha+\beta$}}

 \addlegendimage{empty legend}
\addlegendentry{\hspace{-1.5cm} }

\addplot [color=mycolor1, line width=1.5pt, mark size=2pt, mark=triangle, mark options={solid, rotate=90, mycolor1}]
  table[row sep=crcr]{%
0	0.5\\
1	1\\
2	1\\
3	1\\
4	1\\
5	1\\
6	1\\
7	1\\
8	1\\
9	1\\
10	1\\
11	1\\
12	1\\
13	1\\
14	1\\
15	1\\
16	1\\
17	1\\
18	1\\
19	1\\
20	1\\
};
\addlegendentry{1}

\addplot [color=mycolor2, line width=1.5pt, mark size=3pt, mark=x, mark options={solid, mycolor2}]
  table[row sep=crcr]{%
0	0.5\\
1	0.769230769230769\\
2	0.91743119266055\\
3	0.973709834469328\\
4	0.991965082829084\\
5	0.997575890585876\\
6	0.999271531053862\\
7	0.999781347819232\\
8	0.99993439430439\\
9	0.999980317387413\\
10	0.999994095134868\\
11	0.999998228533138\\
12	0.999999468559282\\
13	0.999999840567725\\
14	0.999999952170312\\
15	0.999999985651093\\
16	0.999999995695328\\
17	0.999999998708598\\
18	0.999999999612579\\
19	0.999999999883774\\
20	0.999999999965132\\
};
\addlegendentry{0.7}

\addplot [color=mycolor3, line width=1.5pt, mark size=2pt, mark=*, mark options={solid, mycolor3}]
  table[row sep=crcr]{%
0	0.5\\
1	0.625\\
2	0.735294117647059\\
3	0.822368421052631\\
4	0.885269121813031\\
5	0.927850356294537\\
6	0.955423749541397\\
7	0.972768702061958\\
8	0.98348129088135\\
9	0.990022850677816\\
10	0.993989724239196\\
11	0.996385144031034\\
12	0.997827945753317\\
13	0.998695634190672\\
14	0.999216971972409\\
15	0.999530035985447\\
16	0.99971796857342\\
17	0.999830762051884\\
18	0.999898450356709\\
19	0.999939067738966\\
20	0.9999634397523\\
};
\addlegendentry{0.4}

\addplot [color=mycolor4, line width=1.5pt, mark size=2pt, mark=diamond, mark options={solid, mycolor4}]
  table[row sep=crcr]{%
0	0.5\\
1	0.555555555555556\\
2	0.609756097560976\\
3	0.661375661375661\\
4	0.709421112372304\\
5	0.753193540612196\\
6	0.792302621570914\\
7	0.82664084901967\\
8	0.856331426842715\\
9	0.881664935499932\\
10	0.903037126829805\\
11	0.920895664738407\\
12	0.9356992379835\\
13	0.947889238046222\\
14	0.957872329434476\\
15	0.966011492215792\\
16	0.972623093734289\\
17	0.977977895569679\\
18	0.982304377486352\\
19	0.985793222434495\\
20	0.988602192725058\\
};
\addlegendentry{0.2}

\end{axis}
\draw[line width=0.45pt] (7.38,5.6) -- (8.9,5.6);
\end{tikzpicture}%

%% file: figure_proof.tex
\begin{tikzpicture}
	
	\node at (3.5,5) [rectangle,rounded corners=3mm,minimum height=1cm,minimum width=6cm, draw, fill=bleudefrance2] {};
	\node at (3.5,3.5) [rectangle,rounded corners=3mm,minimum height=1cm,minimum width=6cm, draw, fill=darkgreen2] {};
	\node at (7.2,2.5) [ rotate=90] {\Large $\dots$};
	\node at (3.5,1.5) [rectangle,rounded corners=3mm,minimum height=1cm,minimum width=6cm, draw, fill=magenta2] {};
	
	\node at (7.2,5)  {$i\!=1$};
	\node at (7.2,3.5)  {$i\!=2$};
	\node at (7.5,1.5)  {$ i\!=\!\ell\! -\! 1$};
	
	
	\draw[bleudefrance,dashed] (2,5.5) -- (2,1);
	\draw[bleudefrance,dashed] (4.5,5.5) -- (4.5,1);
	
	\draw[darkgreen,dashed] (3.5,5.5) -- (3.5,1);
	
	\draw[magenta,dashed] (5,5.5) -- (5,1);
	
	
	\draw[thick,bleudefrance] (2,5.5) -- (2,4.5);
	\draw[thick,bleudefrance] (4.5,5.5) -- (4.5,4.5);
	\node at (1.25,5) {$v_{1,1}$};
	\node at (3.25,5) {$v_{1,2}$};
	\node at (5.5,5) {$v_{1,3}$};
	
	\draw[thick,darkgreen] (3.5,4) -- (3.5,3);
	\node at (2,3.5) {$v_{2,1}$};
	\node at (5,3.5) {$v_{2,2}$};
	
	\draw[thick,magenta] (5,2) -- (5,1);
	\node at (5.75,1.5) {$v_{i,e_{i}}$};
	
	\draw[<->] (0.55,5.75) -- ( 1.95,5.75);
	\draw[<->] (2.05,5.75) -- ( 3.45,5.75);
	\draw[<->] (3.55,5.75) -- ( 4.45,5.75);
	\draw[dotted,line width= 0.3mm] (4.55,5.75) -- ( 4.95,5.75);
	\draw[<->] (5.05,5.75) -- ( 6.45,5.75);
	
	\node at (1.25,6) {$\nu_1$};
	\node at (2.75,6) {$\nu_2$};
	\node at (4,6) {$\nu_3$};
	\node at (5.75,6) {$\nu_e$};
	
\end{tikzpicture}